\documentclass[12pt,oneside,reqno]{amsart}

\usepackage[T1]{fontenc}

\usepackage{amssymb}
\usepackage{array}
\usepackage[utf8]{inputenc}
\usepackage[english]{babel}
\usepackage{caption}
\usepackage[labelfont=md]{subcaption}
\usepackage{enumitem}

\usepackage[a4paper,tmargin=1in,bmargin=1.2in ,lmargin=1in,rmargin=1in,headheight=0in,headsep=0in,footskip=1.5\baselineskip]{geometry}
\usepackage{natbib,nicefrac,setspace,tabularx,graphicx,tikz,float,mathtools,bbm,paralist,amsmath}

\usepackage{verbatim} 

\usepackage[colorlinks=true,linkcolor=blue,citecolor=blue]{hyperref}

\let\oldFootnote\footnote
\newcommand\nextToken\relax

\renewcommand\footnote[1]{%
    \oldFootnote{#1}\futurelet\nextToken\isFootnote}

\newcommand\isFootnote{%
    \ifx\footnote\nextToken\textsuperscript{,}\fi}

\counterwithin{figure}{section}

\newtheorem{theorem}{Theorem}
\newtheorem*{theorem*}{Theorem}
\newtheorem{lemma}{Lemma}
\newtheorem{proposition}{Proposition}
\newtheorem{claim}{Claim}
\newtheorem{corollary}{Corollary}

\theoremstyle{definition}

\newtheorem*{definition*}{Definition}
\newtheorem*{lemma*}{Lemma}
\newtheorem{example}{Example}

\newcommand{\fb}{\mathfrak{b}}
\newcommand{\fg}{\mathfrak{g}}

\newcommand{\eps}{\varepsilon}

\newcommand{\E}{\mathbb{E}}
\newcommand{\Prob}{\mathbb{P}}
\DeclareMathOperator*{\argmax}{arg\,max}

\def\ee{\mathrm{e}}

\makeatletter
\renewenvironment{proof}[1][\proofname]{%
  \par\pushQED{\qed}\normalfont%
  \topsep6\p@\@plus6\p@\relax
  \trivlist\item[\hskip\labelsep\bfseries#1\@addpunct{.}]%
  \ignorespaces
}{%
  \popQED\endtrivlist\@endpefalse
}
\makeatother

\title{Learning about Informativeness}

\author[]{Wanying Huang$^{\dagger}$}
\thanks{
$^{\dagger}$Department of Economics, Monash University. Email: \texttt{kate.huang@monash.edu}.\\
\textit{Acknowledgements.} This work is based on the first chapter of my PhD dissertation at Caltech, and was awarded ``best student paper'' at EC'24. Part of the work was done while the author was supported by a PIMCO Graduate Fellowship in Data Science at Caltech. 
\\
I thank the editor, an associate editor, and two anonymous referees for their constructive comments and suggestions, which substantially improved the paper. I am indebted to my advisor Omer Tamuz for his continued support and encouragement. I would like to thank  (in alphabetical order) Boğaçhan Çelen, Krishna Dasaratha, Laura Doval, Federico Echenique, Amanda Friedenberg, Wade Hann-Caruthers, SangMok Lee, Kirby Nielson, Luciano Pomatto, Fan Wu, and various seminar and conference participants for their helpful comments and suggestions.} 
\date{\today}
\begin{document}

\maketitle

\onehalfspacing

\begin{abstract}
We study whether individuals can learn the informativeness of their information technology through social learning. As in the classic sequential social learning model, rational agents arrive in order and make decisions based on the past actions of others and their private signals. There is uncertainty regarding the informativeness of the common signal-generating process. We show that in this setting asymptotic learning about informativeness is not guaranteed and depends crucially on the relative tail distributions of private beliefs induced by uninformative and informative signals. We identify the phenomenon of perpetual disagreement as the cause of learning and characterize learning in the canonical Gaussian environment.  
\end{abstract}

\section{Introduction}
A key assumption in social learning is that useful information about an unknown state of the world is privately dispersed among individuals. 
Under this assumption, much of the literature has focused on when this information can be aggregated through social learning\footnote{See e.g. \cite{banerjee1992simple, BichHirshWelch:92, smith2000pathological}.} and which network structures facilitate aggregation.\footnote{See e.g. \cite{ccelen2004observational, acemoglu2011bayesian, lobel2015information}.} In practice, however, individuals may have no useful private information about the state at all. For example, in financial markets, noise traders may base their decisions on signals unrelated to an asset’s intrinsic value \citep{banerjee2015signal}. Naive consumers may rely on reviews fabricated independently of product quality, while voters may be swayed by political rumors that reveal nothing about policy outcomes and remain uninformed \citep{kartal2022fake}.
These situations raise a natural question: can society eventually learn whether there is useful private information to aggregate from observed behavior in the first place?

In this paper, we explore how the possibility that agents receive uninformative private signals interferes with social learning. We ask whether an outside observer can eventually learn whether agents' private signals come from an informative or uninformative source from their behavior. We refer to this as \textit{asymptotic learning of informativeness}.
Understanding whether such learning holds is important, particularly from the perspective of a regulator and policymaker whose intervention decisions in many domains\textemdash such as financial markets, online platforms, and elections\textemdash depend on their ability to identify uninformative sources.



Formally, we introduce uncertainty regarding the informativeness of the source into the canonical sequential social learning model \citep*{banerjee1992simple, BichHirshWelch:92, smith2000pathological}: short-lived agents arrive in order, each acting once by choosing a binary action. Agents share a common objective of matching their actions to an unknown binary payoff-relevant state, whose realizations are a priori equally likely. Before making her decision, each agent observes the past actions of her predecessors and receives a private signal from a common source of information. Unlike in the canonical setting, in our model, the informativeness of this source is uncertain. In particular, we assume that the source can be either informative, generating private signals that are independent and identically distributed (i.i.d.) conditional on the payoff-relevant state, or uninformative, producing private signals that are i.i.d.\ but independent of the payoff-relevant state. Both the state and the informativeness of the source are realized independently at the outset and remain fixed throughout.

If the observer were to have access to the private signals received by the agents, he would gradually accumulate empirical evidence about the source and eventually learn its informativeness. Thus, our main focus is on the setting in which  only the history of past actions is observable. In this setting, the observer's inference problem becomes more challenging\textemdash not only because there is less information available, but also because these past actions are correlated with each other. This correlation arises from social learning behavior, where agents' decisions are influenced by the inferences they draw from observing others' actions. We say that \emph{asymptotic learning of informativeness} holds if the observer's belief that the source is informative converges almost surely to one when it is informative and to zero when it is uninformative. 



The notion of learning about the source's informativeness is new to the social learning literature. Absent information uncertainty, a common inquiry is whether agents eventually herd on the correct action\textemdash i.e., the action that matches the payoff-relevant state\textemdash and whether their beliefs also converge to the truth. With unbounded signals, \cite{smith2000pathological} show that both correct herding and convergence of beliefs to the truth occur. In the presence of information uncertainty, however, we show that this is no longer the case: agents may eventually herd on the correct action without ever being sure that it is correct.

We consider \textit{unbounded signals}, under which the agent's private belief induced by an informative signal can be arbitrarily strong. We focus on this setting since otherwise, learning about the source's informativeness can easily be precluded by agents' lack of response to their private signals.\footnote{The phenomenon in which agents follow the actions of their predecessors regardless of their private signals is known as an \emph{information cascade}. In such cascades, since agents do not respond to their private signals, their actions no longer reveal any private information, and thus information stops accumulating. As shown in \cite*{BichHirshWelch:92}, information cascades occur almost surely when signals are bounded and the set of possible signal values is finite.}  Our main result (Theorem \ref{thm:mainthm}) shows that even with unbounded signals, achieving asymptotic learning of informativeness is far from guaranteed. The key factor in determining such learning lies in the tail distributions of agents' private beliefs. In particular, it depends on whether the belief distribution induced by uninformative signals has \emph{fatter} or \emph{thinner} tails compared to that induced by informative signals. We show that asymptotic learning of informativeness holds when uninformative signals have fatter tails than informative signals, but fails when uninformative signals have thinner tails.

For example, consider an informative source that generates Gaussian signals with unit variance and a mean of $+1$ if the payoff-relevant state is good and a mean of $-1$ if the state is bad. Meanwhile, the uninformative source generates Gaussian signals with mean $0$, independent of the state. If the uninformative source generates signals with a variance strictly greater than one, then the uninformative signals have fatter tails, and thus asymptotic learning of informativeness holds. In contrast, when the uninformative signals have a variance strictly less than one, they exhibit thinner tails, and so asymptotic learning of informativeness fails.

The mechanism behind the main result (Theorem \ref{thm:mainthm}) is as follows. With a uniform prior over the state, agents always act as if their signals are informative despite uncertainty about the source. Therefore, when the source is indeed informative and generates unbounded signals, agents will eventually take the correct action. Now, suppose the source is uninformative and generates signals with thinner tails. In this case, it is unlikely that agents will receive signals extreme enough to break the current action consensus, so they usually mimic their predecessors. As a result, an outside observer who only observes agents' actions cannot discern that the source is uninformative, as such a consensus can arise under either source. In contrast, suppose  the source is uninformative but generates signals with fatter tails. In this scenario, agents are more likely to receive signals extreme enough to break the consensus; in fact, both actions will be taken infinitely often, so agents' actions do not converge. Hence, an outside observer who observes an infinite number of action switches learns that the source is uninformative.

More generally, for some private belief distributions, their relative tail thickness is neither thinner nor fatter. For these, we show that the same mechanism holds: asymptotic learning of informativeness is achieved if and only if conditional on the source being uninformative, agents never settle on an action consensus (Proposition \ref{thm:main}).

As mentioned above, learning about the source's informativeness also matters for learning about the payoff-relevant state. In Section \ref{sec:otherlearning}, we show that when the source is informative, although agents eventually all choose the correct action, their social beliefs converge to the truth if and only if they learn almost surely that the source is indeed informative (Proposition \ref{prop:standardlearning}). In contrast, when the source is uninformative, agents are clearly not
guaranteed to reach a consensus on the correct action; in fact, their actions may or may not converge at all. As implied by Proposition \ref{thm:main}, in this case, failure to learn that the source is uninformative is accompanied by an action consensus, which may be incorrect. Consequently, the corresponding social belief also need not converge to the prior (Proposition \ref{prop:standardlearning2}).





An important feature of our model is that, despite uncertainty about the source’s informativeness, agents always act as if their signals are informative and remain responsive to them. This follows from the uniform-prior assumption, which implies that agents are indifferent between the two actions conditional on the source being uninformative. When we relax this assumption, Proposition \ref{prop:nonuniform_fails} shows that asymptotic learning of informativeness fails since having a nonuniform prior induces a default action. When agents become sufficiently convinced that the source is uninformative, they choose this action and stop responding to their signals. In consequence, differences between the sources are no longer reflected in their behavior, causing learning about informativeness to fail. This failure thus highlights the importance of agents’ continued responsiveness to their signals in achieving asymptotic learning of informativeness. As an alternative interpretation, one can also view the behavior implied by the uniform-prior assumption as a form of misspecified learning: agents neglect the possibility that the source is uninformative and always perceive it as informative (see Section \ref{sec:discussion}).

\subsection*{Related Literature}
Our paper contributes to a rich literature on sequential social learning. Assuming that the common source of information is always informative, the primary focus of this literature has been on determining whether agents can eventually learn to choose the correct action. Various factors, such as the information structure \citep*{banerjee1992simple, BichHirshWelch:92, smith2000pathological} and the observational networks \citep*{ccelen2004observational, acemoglu2011bayesian, lobel2015information}, have been extensively studied to analyze their impact on information aggregation, including its efficiency \citep*{rosenberg2019efficiency} and the speed of learning \citep*[e.g.,][]{vives1993fast,hann2018speed}. However, the question of learning about the informativeness of the source\textemdash which is the focus of this paper\textemdash remains largely unexplored.\footnote{For comprehensive surveys on recent developments in the social learning literature, see e.g., \cite*{golub2017learning, bikhchandani2021information}.} 


Our paper also relates to models with multiple information sources and to the broader literature on learning under model uncertainty.\footnote{See \cite{liang2025learning} for a recent experimental study documenting that when updating their beliefs, individuals underreact to reports from sources with uncertain accuracy.} For example, \cite{acemoglu2016fragility} show that small uncertainty about the underlying signal distributions can generate substantial disagreement in the long-run beliefs of individuals with heterogeneous priors. When agents are unsure about one another’s prior beliefs, 
\cite{sethi2016communication} study the tradeoff between choosing a source with more precise signals and one whose prior beliefs are better known.\footnote{Relatedly, \cite{min2026information} shows that small uncertainty about signal structures can lead an agent to eventually acquire information only from a familiar source that is less informative than the unfamiliar alternative.} Focusing on information acquisition, 
\cite*{liang2020complementary} consider a model in which agents endogenously choose from a set of correlated information sources, and the acquired information is then made public to other agents. They show that information complementarities can result in either efficient information aggregation or ``learning traps,'' in which society gets stuck in choosing suboptimal information structures.

In a setting closer to ours, 
\cite*{chen2022sequential} examines a sequential social learning model in which ambiguity-averse agents know their own signal structures but face ambiguity about those of their predecessors. He shows that under sufficient ambiguity aversion, there can be incorrect information cascades even with unbounded signals. Our paper instead considers Bayesian uncertainty about the informativeness of a common source. As we show later in Section \ref{sec:otherlearning}, when the source is informative but uncertainty about its informativeness remains unresolved, agents eventually all choose the correct action but society may remain uncertain, even in the limit, about whether that action is correct. Thus, Bayesian uncertainty can also prevent complete learning, echoing the results from \cite{chen2022sequential}. In a related setting, \cite{xu2025social} shows that when agents arrive in cohorts and observe coarse signals of past cohorts’ actions, unbounded private beliefs alone need not ensure asymptotic learning in probability. 





Another way of viewing our model is by considering a social learning model with four states: the source is either informative with the good or bad state, or uninformative with either the good or bad state. In such multi-state settings, recent work by \cite*{arieli2021general} demonstrates that pairwise unbounded signals are necessary and sufficient for learning, when the decision problem that agents face includes a distinct action that is uniquely optimal for each state. This is not the case in our model, because the same action is optimal in different states, e.g., when the source is uninformative, and so even when agents observe a very strong signal indicating that the state is uninformative, they do not reveal it in their behavior. 

More recently, \cite*{kartik2022beyond} consider a setting with multiple states and actions on general sequential observational networks. They identify a sufficient condition for learning \textemdash ``excludability'' \textemdash that jointly depends on the information structure and agents' preferences. Roughly speaking, this condition ensures that agents can always displace the wrong action, which is their driving force for learning. In our model, when the source is uninformative, agents cannot displace the wrong action as all signals are pure noise.\footnote{This observation can also be seen from Theorem 2 in \cite*{kartik2022beyond}.} Conceptually, our approach differs from theirs as we are interested in identifying the uninformative state from the informative one, instead of identifying the payoff-relevant state.

Our paper is also related to the growing literature on social learning with misspecified models. \cite*{bohren2016informational} investigates a model where agents fail to account for the correlation between actions, demonstrating that different degrees of misperception can lead to distinct learning outcomes. In a broader framework, \cite*{bohren2021learning} show that learning remains robust to minor misspecifications. In contrast, \cite*{frick2020misinterpreting} find that an incorrect understanding of other agents' preferences or types can result in a severe breakdown of information aggregation, even with a small amount of misperception. Later, \cite*{frick2023belief} propose a unified approach to establish convergence results in misspecified learning environments where the standard martingale approach fails to hold. On a more positive note, \cite*{arieli2023hazards} illustrate that by being mildly condescending\textemdash misperceiving others as having slightly lower-quality information\textemdash agents may perform better in the sense that on average, only finitely many of them take incorrect actions.

\section{Model}
There is an unknown binary state of the world $\theta \in \{\fg, \fb\}$, chosen at time $0$ with equal probability. We refer to $\fg$ as the good state and $\fb$ as the bad state. Concurrently, nature chooses an additional binary state $\omega \in \{0, 1\}$, independent of $\theta$, where $\omega=1$ with probability $\gamma \in (0,1)$.\footnote{Our results do not rely on the independence between $\theta$ and $\omega$. They hold true as long as conditioned on $\omega=0$, both realizations of $\theta$ are equally likely.} A countably infinite set of agents indexed by time $t \in \mathbbm{N}=\{1, 2, \ldots\}$ arrive in order, each acting once. The action of agent $t$ is $a_t \in  A =\{\fg, \fb\}$, with a payoff of one if her action matches the state $\theta$ and zero otherwise. Before agent $t$ chooses an action, she observes the history of her predecessors' actions $H_t = (a_1, \ldots, a_{t-1})$ and receives a private signal $s_t$, taking values in a measurable space $(S, \Sigma)$. 

The informativeness of private signals depends on $\omega$, which determines the signal-generating process for all agents. If $\omega=1$, the source is informative, and conditional on $\theta$, signals are drawn i.i.d.\ across agents from a distribution $\mu_{\theta}$. If $\omega=0$, the source is uninformative, and signals are drawn i.i.d.\ from a distribution $\mu_0$ that is independent of $\theta$. We denote by $\Prob_{0}[\cdot] := \Prob[\cdot|~ \omega = 0]$ and $\Prob_{1}[\cdot] := \Prob[\cdot |~ \omega = 1]$ the conditional probability distributions given $\omega = 0$ and $\omega=1$, respectively. Similarly, we use the notation $\Prob_{1, \fg}[\cdot]:= \Prob[\cdot|~ \omega = 1, \theta = \fg]$ to denote the conditional probability distribution given $(\omega, \theta) = (1, \fg)$. We use an analogous notation for $(\omega, \theta) = (1, \fb)$. 

\subsection*{Agents' Behavior} A pure strategy of agent $t$ is a measurable function $\sigma_t: A^{t-1} \times S \to A$ that selects an action for each possible pair of observed history and private signal. A pure strategy profile $\sigma = (\sigma_t)_{t\in \mathbbm{N}}$ is a collection of pure strategies of all agents. A strategy profile is a Bayesian Nash equilibrium\textemdash referred to as equilibrium hereafter\textemdash if no agent can unilaterally deviate from this profile and obtain a strictly higher expected payoff conditioned on their information. 
Given that each agent acts only once, the existence of an equilibrium is guaranteed by a simple inductive argument. In equilibrium, each agent $t$ chooses the action $a_t$ that maximizes her expected payoff given the available information: 
\begin{align}\label{eq:optimal_action_gen}
    a_t \in \argmax_{a \in A} \E[\mathbbm{1}(\theta = a) | H_t, s_t]. 
\end{align}
Below, we make a continuity assumption which implies that agents are never indifferent, and so there is a unique equilibrium.


We first observe that, despite the uncertainty regarding the informativeness of the source, in equilibrium, each agent chooses the action that is most likely to match the state, \emph{conditional} on the source being informative. 
\begin{lemma} \label{claim:opt_action}
    The equilibrium action for each agent $t$ is 
    \begin{equation}
   \label{eq:optimal_action}
    a_t \in \argmax_{a \in A} \Prob_1[\theta = a | H_t, s_t]. 
     \end{equation}
\end{lemma}
That is, agents always act as if signals are informative, irrespective of the underlying signal-generating process. Intuitively, treating signals as informative\textemdash even when they are pure noise\textemdash does not adversely affect agents' payoffs, since in the absence of any useful information, each agent with a uniform prior is indifferent between the available actions. Thus, agents always respond to their private signals. Later, in Section \ref{sec:discussion}, we show that this responsiveness eventually breaks down under a non-uniform prior, and consequently, the source’s informativeness can never be fully revealed.


\subsection*{Information Structure}
We assume that the distributions $\mu_\fg, \mu_{\fb}$, and $\mu_{0}$ are distinct\footnote{Formally, two distributions are distinct if there exists a positive measure set to which they assign a different measure. From an economic perspective, ensuring that $\mu_0$ is distinct from both $\mu_\fg$ and $\mu_\fb$ prevents a situation in which there is a `fake' information technology that perfectly mimics the informative source at one of the realizations of $\theta$ but is actually uninformative. Similarly,  $\mu_\fg$ being distinct from $\mu_\fb$ means that in the informative state, the signals are actually informative.} and mutually absolutely continuous, so no single signal fully reveals either state $\theta$ or $\omega$.
As a consequence, conditioned on $\omega=1$, the log-likelihood ratio of any signal $s$
\[
\ell = \log \frac{d\mu_{\fg}}{d\mu_{\fb}}(s),
\]
is well-defined. We refer to $\ell$ as the \emph{agent's private log-likelihood ratio (LLR)}. Since regardless of the realization of $\omega$, agents always act as if the signals are informative, it is sufficient to consider $\ell$ to capture their behavior.\footnote{Formally, the sequence of actions $a_1, \ldots,  a_t$ is determined by $\ell_1, \ldots, \ell_t$. } Let $F_\fg$ and $F_\fb$ be the cumulative distribution functions (CDFs) of $\ell$ conditioned on $(\omega, \theta) = (1, \fg)$ and $(\omega, \theta) = (1, \fb)$, respectively. Let $F_0$ denote the CDF of $\ell$ conditioned on $\omega =0$. Note that these conditional CDFs are mutually absolutely continuous, as $\mu_\fg, \mu_\fb$ and $\mu_0$ are. Let $f_\fg, f_{\fb}$ and $f_{0}$ denote the corresponding density functions of $F_\fg$, $F_\fb$ and $F_{0}$ whenever they are differentiable. 

We focus on \emph{unbounded signals}, where the agent's private LLR can take on arbitrarily large or small values: for any $M \in \mathbbm{R}$, there exists a positive probability that $\ell > M$ and a positive probability that $\ell < -M$. We informally refer to a signal $s$ as \emph{extreme} when the corresponding $\ell$ it induces has a large absolute value.  A common example of unbounded private signals is the case of Gaussian signals, where $s$ follows a normal distribution $\mathcal{N}(m_{(\omega, \theta)}, \sigma^2)$ with variance $\sigma^2$ and mean $m_{(\omega, \theta)}$ that depends on the pair of states $(\omega, \theta)$. An extreme Gaussian signal is a signal that is, for example, five standard deviations away from its mean $m_{(\omega, \theta)}$.

We make two assumptions for expository simplicity. First, we assume that the pair $(F_{\fg}, F_{\fb})$ of informative conditional CDFs is symmetric around zero, i.e., $F_{\fg}(x) + F_{\fb}(-x) = 1$. Given that the prior on $\theta$ is uniform, this assumption is equivalent to requiring our model to be invariant with respect to a relabeling of the payoff-relevant state. Second, we assume that all conditional CDFs\textemdash $F_{\fg}$, $ F_{\fb}$, and $F_0$\textemdash are continuous, so agents are never indifferent between actions.

In addition, we assume that $F_{\fb}$ has a differentiable left tail, i.e., it is differentiable for all $x$ negative enough, and its probability density function $f_{\fb}$ satisfies the condition that $\limsup_{x\to \infty} f_{\fb}(-x) < 1$. By symmetry, this implies that $F_{\fg}$ also has a differentiable right tail and its density function $f_{\fg}$ satisfies the condition that $\limsup_{x\to \infty} f_{\fg}(x)<1$. This is a mild technical assumption that holds for every non-atomic distribution commonly used in the literature, including the Gaussian distribution. It holds, for instance, whenever the density tends to zero at infinity.

\subsection*{Asymptotic Learning of Informativeness} We now define our notion of asymptotic learning. Let $q_t:= \Prob[\omega=1|H_t]$ be the belief that an outside observer assigns to the source being informative after observing the history of agents' actions from time $1$ to $t-1$. As this observer collects more information over time, his belief $q_t$ converges almost surely since it is a bounded martingale. Formally, we say that \emph{asymptotic learning of informativeness} holds if  for all $\omega \in \{0,1\}$, 
\[
\lim_{t\to \infty} q_t  = \omega \quad \Prob_\omega\text{-almost surely.}
\]    
That is, the outside observer's belief eventually converges to the truth. As we explain below in Section \ref{sec:benchmark}, when all signals are publicly observable, asymptotic learning of informativeness holds regardless of the underlying information structure. In Section \ref{sec:otherlearning}, we show how this notion connects to other existing concepts of learning.

\section{Relative Tail Thickness}\label{sec:fat_thin_tails}
To study the conditions for achieving asymptotic learning of informativeness, it is important to understand the concept of relative tail thickness, which captures the relative likelihood of generating extreme signals from different sources. For any pair of CDFs $(F_0, F_\theta)$ where $\theta \in \{\fg, \fb\}$ and some $x \in \mathbbm{R}_+$, we denote their corresponding ratios by 
\[
L_\theta(x) := \frac{F_0(-x)}{F_\theta(-x)} \quad \text{and} \quad R_\theta(x) := \frac{1-F_0(x)}{1-F_\theta(x)}. 
\]
For large $x$, these represent the left and right tail ratios of $F_0$ over $F_\theta$. 

Formally, we say the uninformative signals have \emph{fatter tails} than the informative signals if there exists an $\eps >0$ such that 
    \begin{align*}
        L_\fb( x) & \geq \eps \quad \text{for all $x$ large enough,}  \\
       \text{and} \quad  R_\fg( x) & \geq \eps \quad  \text{for all $x$ large enough.}
    \end{align*}
We say the uninformative signals have \emph{thinner tails} than the informative signals if there exists an $\eps >0$ such that 
    \begin{align*}
        \text{either} \quad  L_\fg(x) & \leq 1/\eps \quad \text{for all $x$ large enough,}\\
        \text{or} \quad  R_\fb(x) & \leq 1/\eps \quad \text{for all $x$ large enough.}
    \end{align*}


That is, for the uninformative signals to have fatter tails, both their corresponding left \emph{and} right tail ratios must eventually be bounded from below. Conversely, for the uninformative signals to have thinner tails, at least one of the tail ratios\textemdash either left \emph{or} right\textemdash must eventually be bounded from above. In other words, for some small $\eps>0$, the fatter-tail condition means that uninformative tails are not much thinner than the corresponding informative tails, while the thinner-tail condition means that they are not much fatter on at least one side.\footnote{These relative tail conditions do not impose a particular functional form on the distributions. In particular, they do not require private LLRs to have exponential tails of order $e^{-\alpha x}$ for some $\alpha>0$, as in a class of signals commonly studied in the literature \cite[see, e.g.,][]{rosenberg2019efficiency, arieli2023hazards}. As illustrated below, our notion also applies to Gaussian signals, whose private LLR tails decay faster than exponentially.} 

The ``and'' condition is important for defining fatter tails, as it ensures that extreme signals of both types (positive and negative) occur frequently enough to prevent agents from reaching an action consensus. In contrast, the ``or'' condition for thinner tails requires only the infrequency of at least one type of extreme signal, thus allowing agents to reach a consensus with positive probability. As we will discuss in Section \ref{sec:results}, determining whether it is possible for agents to reach a consensus under an uninformative source is key to understanding whether asymptotic learning of informativeness holds.

Note that the first condition $L_\fb(x) \geq \eps$ implies that $L_\fg(x) \geq \eps$. This follows from the well-known fact that $F_{\fg}$ exhibits first-order stochastic dominance over $F_{\fb}$, i.e., $F_{\fg}(x) \leq F_{\fb}(x)$ for all $x \in \mathbb{R}$ \citep*[see, e.g.,][]{smith2000pathological, chamley2004rational, rosenberg2019efficiency}. Similar statements apply to the remaining three conditions. Furthermore, notice that the uninformative signals cannot have both fatter and thinner tails simultaneously, as $F_\fg$ and $F_\fb$ represent distributions of the agent's private log-likelihood ratio.\footnote{In particular, $F_\fg$ and $F_\fb$ satisfy the following inequality: $\ee^x F_\fg(-x)\leq F_\fb(-x)$, for all $x \in \mathbbm{R}$.} However, there are distributions under which the uninformative signals have neither fatter nor thinner tails. For these cases, we characterize the conditions for asymptotic learning of informativeness in the canonical Gaussian environment (see Section \ref{sec:normal}).

Intuitively, compared to informative signals, uninformative signals with fatter tails are more likely to exhibit extreme values. Thus, by Bayes' Theorem, observing an extreme signal suggests that the source is uninformative. In contrast, uninformative signals with thinner tails tend to exhibit moderate values, so observing an extreme signal in this case suggests that the source is informative. Next, we provide three examples of uninformative signals with either fatter or thinner tails.

\begin{figure}[t!]
    \centering
    \begin{subfigure}{0.45\linewidth}
        \includegraphics[width=\linewidth]{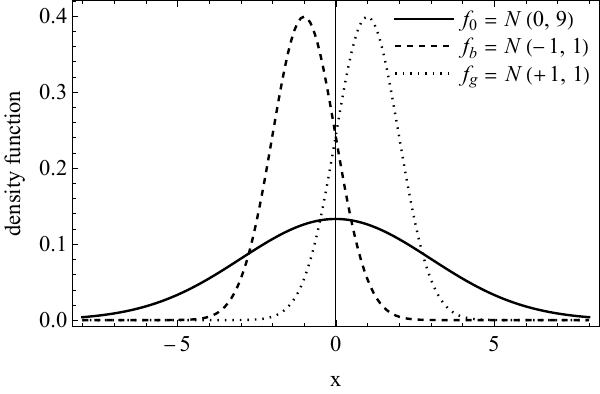}
    \end{subfigure}
    \hfill
    \begin{subfigure}{0.45\linewidth}
        \includegraphics[width=\linewidth]{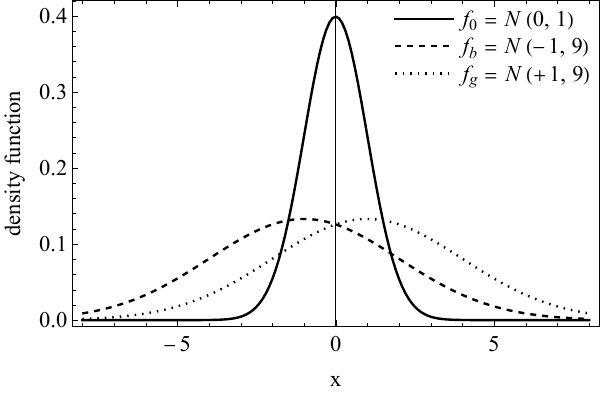}
    \end{subfigure}
    \caption{Probability density functions of informative Gaussian signals and their corresponding uninformative Gaussian signals with fatter tails (on the left) and with thinner tails (on the right).}
    \label{fig:fat_example}
\end{figure}

\begin{example}[Gaussian Signals] \label{ex:gaussian}

Consider the case where $F_\fg$ is normal with mean $+1$ and unit standard deviation and $F_\fb$ is also normal with mean $-1$ and unit standard deviation. 

Suppose that $F_0$ has zero mean. If it has a standard deviation of 3, then the uninformative signals have fatter tails. In this scenario, observing an extreme signal,  such as anything around 11, indicates that it is much more likely that the source is uninformative\textemdash where the signal is less than four standard deviations away\textemdash than that an informative ten-standard-deviation signal is generated under $F_\fg$. Conversely, if the standard deviation of $F_0$ is $1/3$, the uninformative signals have thinner tails, and thus an extreme signal suggests that the source is informative. A graphical example of uninformative Gaussian signals with fatter and thinner tails is shown in Figure \ref{fig:fat_example}.
\end{example}

\begin{example}[First-Order Stochastically Dominated (FOSD) Signals] \label{ex:fosd_thin}
Suppose that $F_0$ first-order stochastically dominates $F_{\fg}$\textemdash i.e.,  $F_0(x) \leq F_\fg(x)$ for all $x \in \mathbbm{R}$. In other words, the uninformative signals are more likely to exhibit high values than the informative signals associated with the good state. This immediately implies that  $L_\fg(x)\leq 1$ for all $x \in \mathbbm{R}$, so the uninformative signals have thinner tails than the informative ones. 

Now suppose an extreme positive signal is observed. It is highly unlikely that the source is informative and associated with the bad state. Instead, it suggests that the source is more likely to be either uninformative or informative but associated with the good state. However, one cannot exclude either possibility based on such an observation, so it does not provide conclusive evidence about the source’s informativeness.
\end{example}


\begin{example}[Mixture  Signals] \label{ex:mixture}
For any pair of distributions $(F_\fg, F_\fb)$ and $\alpha \in (0,1)$, let $F_0 = \alpha F_{\fg} + (1-\alpha) F_{\fb}$. Observe that the uninformative signals represented by $F_0$ have fatter tails than the informative signals.\footnote{To see this, fix any constant $\alpha \in (0,1)$ and let $F_0 = \alpha F_{\fg} + (1-\alpha)F_{\fb}$. Since CDFs always take nonnegative values, for any $x \in \mathbb{R}$, $F_{0}(x) \geq (1-\alpha)F_{\fb}(x)$. Similarly, $1-F_0(x) = \alpha(1-F_{\fg}(x)) + (1-\alpha)(1-F_{\fb}(x)) \geq \alpha (1-F_{\fg}(x))$. Let $\eps = \min\{\alpha, 1-\alpha\}$. By definition, $F_0$ has fatter tails.} In particular, when $\alpha = 1/2$, the corresponding mixture distribution $F_0 = (F_\fg+F_\fb)/2$ coincides with the unconditional distribution of $\ell$ under an informative source. Thus, one can think of it as representing an uninformative source that is a priori indistinguishable from the informative one.

Alternatively, the mixture distribution $F_0$ can also be viewed as generating conditionally i.i.d.\ signals, but conditioned on a different state $\eta$ rather than on $\theta$. Specifically, in each period a new state $\eta$ is randomly drawn from the same set $\{\fg, \fb\}$, independent of $\theta$. With probability $\alpha$, we have $\eta = \fg$ and a signal is drawn from $F_\fg$; with the complementary probability, $\eta = \fb$ and a signal is drawn from $F_\fb$. Thus, the mixture distribution provides a way to construct an uninformative source directly from the informative ones.
\end{example}

\section{Main Results}
\subsection{A Benchmark} \label{sec:benchmark}
As a benchmark, we briefly discuss the case where all signals are observed by the outside observer.\footnote{Equivalently, one can let the outside observer observe all agents' actions in addition to their signals. Since actions contain no additional payoff relevant information, it suffices to only consider the signals.} Depending on the realizations of $\theta$ and $\omega$, these signals are distributed according to either $\mu_\fg$, $\mu_\fb$ or $\mu_0$.  Given that these measures are distinct, as the sample size grows, this observer eventually learns which distribution is being sampled. Formally, at time $t$, the empirical distribution of the signals $\hat{\mu}_t$ assigns to a measurable set $B \subseteq S$ the probability  
\[
\hat{\mu}_t(B) := \frac{1}{t} \sum_{\tau =1}^t \mathbbm{1}(s_\tau \in B).
\]
Conditional on both states $\omega$ and $\theta$, this is the empirical mean of i.i.d.\ Bernoulli random variables. Hence, by the strong law of large numbers, for every measurable set $B \subseteq S$,
\[
\lim_{t \to \infty} \hat{\mu}_t(B) = \mu_{(\omega, \theta)}(B) \quad \text{almost surely,} 
\]
where $\mu_{(1, \fg)} = \mu_\fg$, $\mu_{(1, \fb)}= \mu_\fb$ and $\mu_{(0, \cdot)} = \mu_0$.

Thus, regardless of the underlying signal-generating process, any uncertainty concerning the informativeness of the source is eventually resolved when all signals are publicly observable. 

\subsection{Public Actions} \label{sec:results}
We now turn to our main setting, where signals remain private and only actions are observed. In contrast to the public signal benchmark, our main result (Theorem \ref{thm:mainthm}) shows that achieving asymptotic learning of informativeness is no longer guaranteed. Instead, the key determinant of learning is the relative tail thickness between the uninformative and informative signals, introduced in Section \ref{sec:fat_thin_tails}.

\begin{theorem} \label{thm:mainthm}
When the uninformative signals have fatter (thinner) tails than the informative signals, asymptotic learning of informativeness holds (fails).
\end{theorem}

This result demonstrates that an observer learning from agents' actions can eventually distinguish between informative and uninformative sources, but only if the uninformative source generates signals with greater dispersion. In contrast, when the uninformative signals are relatively concentrated, this differentiation becomes impossible.

For example, consider informative signals that follow a normal distribution with unit variance and mean $+1$ and $-1$, respectively. When the uninformative signals follow a normal distribution with a variance of 2, they have fatter tails and thus Theorem \ref{thm:mainthm} implies that asymptotic learning of informativeness holds. In contrast, when the uninformative signals follow a normal distribution with a variance of 1/2, they have thinner tails, and thus Theorem \ref{thm:mainthm} implies that asymptotic learning of informativeness fails. Indeed, for normal distributions, the relative thickness of the tails is determined solely by their variances: a higher variance corresponds to fatter tails, while a lower variance corresponds to thinner tails (see Lemma \ref{lem:variance_tails} in the appendix). An immediate consequence of Theorem \ref{thm:mainthm} is the following result. 
\begin{corollary} \label{cor:gaussian_variance}
    Suppose all private signals are normal where the informative signals have variance $\sigma^2$, and the uninformative signals have variance $\tau^2$. Then, asymptotic learning of informativeness holds if $\tau > \sigma$ and fails if $\tau < \sigma$.
\end{corollary}

The idea behind our proof of Theorem \ref{thm:mainthm} is as follows. First consider the case where the source is informative. In this case, the likelihood of generating extreme signals that overturn a long streak of correct consensus decreases rapidly. Consequently, agents will eventually reach a consensus on the correct action since they always treat signals as informative. Now, suppose that the source is uninformative and instead of reaching a consensus, agents continue to disagree indefinitely, leading to both actions being taken infinitely often. If this were the case, an outside observer would eventually be able to distinguish between informative and uninformative sources, as they induce distinct behavioral patterns among agents. Whether these disagreements persist or not depends on whether the tails of the uninformative signals are thick enough to generate these overturning extreme signals. 

In sum, when the tails of uninformative signals are sufficiently thick (i.e., when they have fatter tails), overturning extreme signals occur often enough so that disagreements persist. When the tails are not sufficiently thick (i.e., when they have thinner tails), these signals are less likely, so disagreements eventually cease. Hence, we conclude that the relative tail thickness between uninformative and informative signals plays an important role in determining the achievement of asymptotic learning of informativeness.

\subsubsection*{Perpetual Disagreement}  As mentioned above, if agents never reach an action consensus under an uninformative source, then intuitively the outside observer can infer the uninformativeness of the source. To formally establish the key mechanism underlying our main result, we define the event in which agents never converge to any action as \emph{perpetual disagreement}. Let $S := \sum_{t= 1}^\infty \mathbbm{1}(a_t \neq a_{t+1})$ denote the total number of action switches, so that the event $\{S = \infty\}$ is the perpetual disagreement event. It turns out that perpetual disagreement under an uninformative source is not only a sufficient condition for asymptotic learning but also a necessary condition, as shown in the following proposition.


\begin{proposition}\label{thm:main}
Asymptotic learning holds if and only if conditioned on $\omega=0$, the perpetual disagreement event occurs almost surely.
\end{proposition} 
The proof of Proposition \ref{thm:main} uses the idea of an outside observer whose goal is to guess the informativeness of the source. Note that the achievement of asymptotic learning means that eventually he can always make the correct guess. Since the outsider observes all agents' actions, he expects to see an eventual consensus when the source is informative. Therefore, if he observes action nonconvergence, he infers that the source must be uninformative and makes the correct guess. Conversely, suppose to the contrary that agents could also reach a consensus when the source is uninformative. This implies that action convergence is plausible under both informative and uninformative sources. Consequently, the observer is no longer sure of the source's informativeness, contradicting our hypothesis that he always makes the correct guess.

\subsection{Gaussian Private Signals} \label{sec:normal}
We now turn to the canonical environment in which all signals are normally distributed. As shown in Corollary \ref{cor:gaussian_variance}, Theorem \ref{thm:mainthm} applies directly whenever the variance of the uninformative signals, $\tau^2$, differs from that of the informative signals, $\sigma^2$. When these variances coincide and the uninformative mean lies strictly between the two informative means, neither tail condition applies.\footnote{To see this, suppose the uninformative signals have mean $m_0\in(-1,1)$ and the informative signals have means $1$ and $-1$, all with the same variance. For the induced private
LLR distributions, as $x\to\infty$, both $F_0(-x)$ and $F_\fb(-x)$ approach zero, but the former does so much faster than the latter, leading to $L_\fb(x)\to0$. Hence, $F_0$ does not have fatter tails. Similarly, both $L_\fg(x)$ and $R_\fb(x)$ tend to infinity as $x\to\infty$, so $F_0$ does not have thinner tails either.}
The next proposition shows that asymptotic learning of informativeness nevertheless fails in this case. Together with Corollary \ref{cor:gaussian_variance}, this gives a complete characterization for Gaussian signals with $m_0\in(-1,1)$: asymptotic learning of informativeness holds if and only if $\tau>\sigma$.

\begin{proposition}\label{prop:gaussian}
Suppose all private signals are Gaussian with the same variance, where informative signals have means $-1$ and $+1$, and uninformative signals have mean $m_0 \in (-1, 1)$. Then asymptotic learning of informativeness fails. 
\end{proposition}
The intuition is similar to that behind our main result. Although the uninformative signals do not satisfy our thinner-tail condition, they can still sustain an action consensus. As agents observe a longer sequence of identical actions, overturning the consensus requires increasingly extreme private signals. In the equal-variance case, such signals become sufficiently rare that consensus persists with positive probability. Since the same action history has positive probability under either source, the observer cannot always distinguish between them, even in the limit.


\subsection{Numerical Simulation}\label{sec:numerical} To further illustrate our main result, we use Monte Carlo simulations to numerically examine the belief process of an outside observer and the corresponding action switches among agents. We fix the pair of informative Gaussian signal distributions to be $(\mu_\fg, \mu_\fb)$ where $\mu_\fg = \mathcal{N}(+1, 2)$ and $\mu_\fb = \mathcal{N}(-1, 2)$. These processes are simulated under two scenarios: one with fatter-tailed uninformative Gaussian signals with distribution $\mu_0 = \mathcal{N}(0, 3)$, and the other with thinner-tailed uninformative Gaussian signals with distribution $\mu_0 = \mathcal{N}(0, 1)$.\footnote{Note that in this case, the agent's private log-likelihood ratio $\ell = s$, so $F_0, F_\fg$ and $F_\fb$ have the same distribution as $\mu_0$, $\mu_\fg$ and $\mu_\fb$, respectively.} For each uninformative source, we conduct these simulations 1,000 times, redrawing the payoff-relevant state each time, and calculate the averages for each period. This yields approximations for the expected belief and the expected total number of action switches in the presence of uninformative signals. Figure \ref{fig:num_sim} displays the results of these simulations.


\begin{figure}[t]  
     \centering 
 \includegraphics[width=\textwidth]{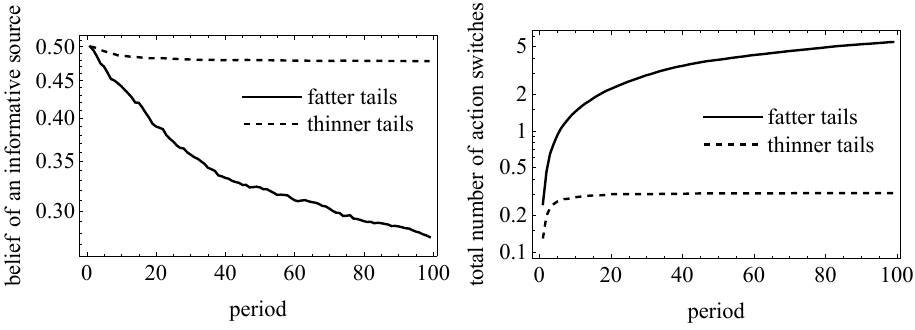}
        \caption{The belief that the source is informative (on the left) and the total number of action switches (on the right) under uninformative signals with fatter and thinner tails.} 
        \label{fig:num_sim}
\end{figure}

What immediately stands out is that under fatter-tailed uninformative signals, the belief of the outside observer that the source is informative decreases much faster compared to thinner-tailed uninformative signals. Under fatter-tailed uninformative signals, by period 60,  this belief is approximately 0.3, which is less than two-thirds of the belief observed under thinner-tailed uninformative signals. These findings align with the predictions of Theorem \ref{thm:mainthm}, suggesting that in the former case the observer will eventually learn that the source is uninformative. In contrast, with thinner-tailed uninformative signals, the decline in the belief about the source's informativeness plateaus around 0.48 after period 20, suggesting that the observer will not be able to learn that the source is uninformative.  

As discussed earlier, the intuition behind Theorem \ref{thm:mainthm} is that the uninformative signals with fatter tails have a higher probability of being extreme, which in turn prevents agents from converging to a consensus. This phenomenon is evident in the right plot of Figure \ref{fig:num_sim}, where the average total number of action switches under fatter-tailed uninformative signals increases over time. In contrast, under the thinner-tailed uninformative signals, the total number of switches plateaus in a short amount of time, suggesting that agents eventually reach an action consensus.

\section{Connections to Other Notions of Learning} \label{sec:otherlearning}
In this section, we discuss the relationship between our notion of asymptotic learning of informativeness and other notions of learning in the literature. One common notion, which concerns the convergence of agents' actions, is known as \textit{herding}. Another notion, which concerns whether agents' beliefs about the payoff-relevant state converge to the truth, is known as \textit{complete learning}. Both of these notions focus on the payoff-relevant state. We show that, in the presence of information uncertainty, agents may eventually choose the correct action while remaining uncertain, even in the limit, about whether it is correct.

Formally, we say that \emph{correct herding} holds if $\lim_{t \to \infty} a_t = \theta$ almost surely. That is, the event that all but finitely many agents take the correct action occurs with probability one. Denote the \textit{social belief} of agent $t$ by $$p_t = \Prob[\theta=\fg \mid H_t].$$ This is the belief assigned by agent $t$ to the good state, given the history
of past actions $H_t= (a_1, \ldots, a_{t-1})$. We say that \emph{complete learning} holds if agents' social beliefs converge to the truth, i.e., $p_t$ almost surely converges to one when $\theta=\fg$ and to zero when $\theta=\fb$. Essentially, complete learning means that society becomes arbitrarily confident about the payoff-relevant state and hence about which action is correct.

In the standard model without information uncertainty, \cite{smith2000pathological} show that if the source generates informative and unbounded signals, then both correct herding and complete learning hold. That is, agents will eventually choose the correct action, and their beliefs will converge to the truth. Our results show that information uncertainty can break this connection: correct herding may hold even when complete learning fails. 

Consider first the case in which the source is informative. In this case, since agents in our model always behave as if the source were informative (Lemma \ref{claim:opt_action}) and informative signals are unbounded, correct herding holds: 
\[
\lim_{t \to \infty} a_t = \theta, \quad \Prob_1\text{-almost surely.}
\]

Thus, regardless of whether asymptotic learning of informativeness holds, agents eventually choose the correct action when the source is informative. However, correct herding does not necessarily imply complete learning.  Although agents eventually take the correct action, their social beliefs converge to the truth only if they also learn that the source is indeed informative, as shown in the following proposition.  

\begin{proposition} \label{prop:standardlearning}
    Conditioned on $\omega=1$, asymptotic learning of informativeness holds if and only if complete learning holds.
\end{proposition}
The proof of Proposition \ref{prop:standardlearning} considers an outsider who observes the same social information as agent $t$ but not her private signal. Thus, at time $t$, the observer's belief about the payoff-relevant state equals the social belief. As long as the observer remains uncertain about the source's informativeness, he cannot fully rely on the action history to infer the payoff-relevant state, since it may have been generated by an uninformative source. Conversely, as the observer becomes increasingly certain that the source is informative, the action history allows him to learn the payoff-relevant state, and so his belief converges to the truth.

Next, consider the case in which the source is uninformative. In this case, agents are clearly not guaranteed to choose the correct action; in fact, their actions may or may not converge at all. Indeed, as shown in Proposition \ref{thm:main}, agents' actions do not converge if asymptotic learning of informativeness holds, or they may converge to the wrong action if asymptotic learning of informativeness fails. The next proposition gives the corresponding implication for social beliefs. It shows that if society learns that the source is uninformative, then the social belief returns to the uniform prior; otherwise, an action consensus may lead society to favor one state even though the source is uninformative.

\begin{proposition} \label{prop:standardlearning2}
    Conditioned on $\omega=0$, asymptotic learning of informativeness holds if and only if $\lim_{t\to \infty} p_t = 1/2$  almost surely.
\end{proposition}

 The proof of Proposition \ref{prop:standardlearning2} follows a similar approach to Proposition \ref{prop:standardlearning} and applies the result of Proposition \ref{thm:main}. One direction is  straightforward: if the observer learns that the source is uninformative, then the social belief must converge to the prior as the action history contains no information about the payoff-relevant state. For the other direction, suppose by contraposition that the observer does not learn that the source is uninformative. Then, Proposition \ref{thm:main} implies that agents may eventually reach a consensus on some action. Since this is also possible under an informative source, the observer cannot fully dismiss the possibility that agents' actions are informative. As a result, his belief about the payoff-relevant state converges to a value different from the prior with positive probability.

\section{Analysis}
In this section we first examine how agents update their beliefs. We present two standard yet useful properties of their belief updating process, namely, the \emph{overturning principle} and \emph{stationarity}. We then use two important phenomena\textemdash immediate agreement and perpetual disagreement\textemdash to characterize asymptotic learning of informativeness. This leads to a proof sketch of Theorem \ref{thm:mainthm} at the end of this section.

\subsection{Agents' Belief Dynamics}  \label{sec:belief_dynamics}
Since agents always act as if signals are informative (Lemma \ref{claim:opt_action}), we focus on how agents update their beliefs conditioned on an informative source. Let $\pi_t := \Prob_1[\theta=\fg|H_t]$ denote agent $t$'s public belief conditional on the action history $H_t$ and $\omega=1$. Its log-likelihood ratio (LLR) is
$$r_t := \log \frac{\pi_t}{1-\pi_t} = \log \frac{\Prob_1[\theta=\fg|H_t]}{\Prob_1[\theta=\fb|H_t]}.$$
Similarly, let 
\[
L_t := \log \frac{\Prob_1[\theta = \fg|H_t, s_t]}{\Prob_1[\theta = \fb|H_t, s_t]}
\]
be the LLR of agent $t$'s posterior belief, conditional on $H_t$, her private signal $s_t$ and $\omega=1$. Recall that $\ell_t$ is agent $t$'s private LLR conditional on $\omega=1$. By Bayes' rule, we have  
\[
L_t = r_t + \ell_t.
\]
From \eqref{eq:optimal_action} we know that in equilibrium, agent $t$ chooses action $\fg$ if $\ell_t \geq -r_t$ and action $\fb$ if $\ell_t <-r_t$. Hence, conditioned on $\theta$ and the event $\omega=1$, the probability that $a_t =\fg$ is $1-F_{\theta}(-r_t)$ and the probability that $a_t= \fb$ is $F_\theta(-r_t)$. Consequently, the agents' public LLRs evolve as follows:
\begin{align}
    r_{t+1} = r_t + D_{\fg}(r_t) \quad \text{if~} a_t = \fg, \label{eq:highaction}\\
     r_{t+1} = r_t + D_{\fb}(r_t)  \quad \text{if~} a_t = \fb, \label{eq:lowaction}
\end{align}
where 
\begin{align*}
  D_{\fg}(r) := \log \frac{1-F_{\fg}(-r)}{1-F_{\fb}(-r)} \quad \text{and} \quad   D_{\fb}(r) := \log \frac{F_{\fg}(-r)}{F_{\fb}(-r)}.
\end{align*}
Note that $D_{\fg}$ always takes positive values and $D_{\fb}$ always takes negative values as $F_\fg$
first-order stochastically dominates $F_\fb$.

\subsubsection*{Overturning Principle and Stationarity.}

One important property held by the agent's
public belief is known as the \textit{overturning principle} \citep*{sorensen1996rational, smith2000pathological}, which asserts that observing a single action switch is sufficient to change the verdict of $\pi_t$.  


\begin{lemma}[Overturning Principle]
\label{claim:overturning}
    For each agent $t$, if $a_{t} = \fg$, then $\pi_{t+1} \geq 1/2$. Similarly, if $a_{t} = \fb$, then $\pi_{t+1} \leq 1/2$. 
\end{lemma}

Another important property of $\pi_t$ is \emph{stationarity}\textemdash conditional on a fixed pair $(\omega,\theta)$, the value of $\pi_t$ captures all past information relevant to subsequent actions, independently of time. This holds in our model because, regardless of the informativeness of the source, agents always update their public log-likelihood ratios according to either \eqref{eq:highaction} or \eqref{eq:lowaction}. Future private signals are i.i.d.\ and independent of past history, so these updating rules and the current belief $\pi_t$, together with agents’ decision rule, fully determine the conditional distribution of subsequent actions.
We further write $\Prob_{\tilde\omega, \tilde\theta, \pi}$ to denote the conditional probability distribution given the pair of state realizations ($\tilde\omega, \tilde\theta$) while highlighting the different values of the prior $\pi$. 

\begin{lemma}[Stationarity]\label{claim:stationary}
For any fixed sequence $(b_\tau)_{\tau=1}^k$ of $k$ actions in $\{\fg, \fb\}$, any prior $\pi \in (0,1)$ and any pair $(\tilde{\omega}, \tilde{\theta})  \in \{0, 1\} \times \{\fg, \fb\}$
\[
\Prob_{\tilde\omega, \tilde\theta}[a_{t} = b_1, \ldots, a_{t+k-1}=b_k|\pi_t = \pi] = \Prob_{\tilde\omega, \tilde\theta, \pi}[a_{1} = b_1, \ldots, a_{k}=b_k].
\]
\end{lemma}
This lemma states that holding the states fixed, if agent $t$'s public belief is equal to $\pi$, then the probability of observing a sequence $(b_1, \ldots, b_k)$ of actions of length $k$ is the same as observing this sequence starting from time $1$, with initial public belief $\pi$. 


\subsection{ Asymptotic Learning of Informativeness and Immediate Agreement}

In our model, since the agent's public belief $\pi_t$ evolves as in the standard model, it remains a martingale when the source is informative. However, when the source is uninformative, $\pi_t$ ceases to be a martingale under the measure $\Prob_0$. Therefore, we need to employ a different analytical approach to understand what ensures asymptotic learning.\footnote{A similar approach can be found in \cite*{arieli2023hazards}, where the agent's public belief also ceases to be a martingale under the correct measure because overconfident agents have misspecified beliefs.} To do so, we focus on the event $\{a_1 = a_2 = \ldots = a\}$ in which all agents immediately reach a consensus on some action $a \in \{\fg, \fb\}$. We denote such an event by $\{\bar{a} = a\}$ and refer to it as \emph{immediate agreement on action $a$}. Notice that conditioned on this event, the process of public log-likelihood ratios $(r_t)$ is deterministic and evolves according to either \eqref{eq:highaction} or \eqref{eq:lowaction}. 

The following lemma shows that when the source is informative, immediate agreement on the wrong action is impossible, whereas immediate agreement on the correct action is possible, at least for some prior. For brevity, we state this result only for the case where $\theta=\fg$. By symmetry, analogous statements hold for $\theta =\fb$.

\begin{lemma}\label{lem:informative} Conditioned on $\omega = 1$ and $\theta= \fg$, the following two conditions hold: 
\begin{compactenum}[(i)]  
    \item The event $\{\bar{a} = \fb \}$ occurs with probability zero.   
    \item The event  $\{\bar{a} = \fg \}$ occurs with positive probability for some prior $\pi \in (0,1)$.
    \end{compactenum}
\end{lemma} 


The first part of Lemma \ref{lem:informative} holds since, in our model, conditioned on an informative source, all but finitely many agents take the correct action. This immediately implies that agents cannot reach a consensus on the wrong action from the outset. Likewise, the second part of Lemma \ref{lem:informative} holds because if agents eventually reach a consensus on the correct action, by stationarity, they can also do so immediately at least for some prior.\footnote{In fact, part (ii) of Lemma \ref{lem:informative} holds not only for some prior but also for any prior. This can be seen by applying a similar argument used in the proof of Lemma \ref{prop:pos_herd} in the appendix. We omit the stronger statement here, as it is not required to prove our main result.}

Next, we focus on the immediate agreement event conditioned on an uninformative source. Recall that in Proposition \ref{thm:main}, we establish the relationship between learning about informativeness and perpetual disagreement. Building on this, we now characterize asymptotic learning of informativeness in terms of immediate agreement, which is a crucial step in proving our main result. 

\begin{proposition}
\label{cor:immd_learn_uniform}
Asymptotic learning of informativeness holds if and only if conditioned on $\omega=0$, both events $\{\bar a= \fg\}$ and $\{\bar a= \fb\}$ occur with zero probability.
\end{proposition}
This proposition states that achieving asymptotic learning of informativeness is equivalent to the absence of immediate agreement starting from a uniform prior given an uninformative source. The proof of Proposition \ref{cor:immd_learn_uniform} utilizes the idea that the agent's belief updating process is eventually monotonic\textemdash a technical property that we establish in Lemma \ref{lem:eventual_mon} in the appendix. This property ensures that if immediate agreement is impossible for some prior, e.g., the uniform prior, it becomes impossible for \emph{all} priors. By stationarity, this then implies that agents cannot reach a consensus on any actions for any prior. Consequently, perpetual disagreement occurs with probability one, and it follows from Proposition \ref{thm:main} that this is equivalent to achieving asymptotic learning of informativeness.

Therefore, we have reduced the problem of asymptotic learning of informativeness to determining the possibility of immediate agreement conditioned on an uninformative source, which is much easier to analyze. Specifically, conditioned on the event $\{\bar{a} = \fg\}$, let $r^{\fg}_t$ denote the deterministic process of $r_t$ based on \eqref{eq:highaction}. Recall that agent $t$ chooses action $\fg$ if $\ell_t \geq -r_t$ and action $ \fb$ otherwise. Consequently, the probability of $\{\bar{a} = \fg\}$ is equal to the probability that $\ell_t \geq -r^{\fg}_t$ for all $t \geq 1$. Moreover, conditioned on an uninformative source, since private signals are i.i.d., the corresponding private log-likelihood ratios are also i.i.d.\ Thus, conditioned on $\omega=0$, the probability of immediate agreement on action $\fg$ is  
\[
\Prob_0[\bar{a}=\fg] = \prod_{t=1}^{\infty}(1-F_0(-r^{\fg}_t)).
\] 
To determine whether the above probability is positive or zero, by a standard approximation argument, it is equivalent to examining whether the sum of the probabilities of the following events is finite or infinite:
\begin{equation} \label{eq:sum_infinite1}
 \Prob_0[\bar{a}=\fg] > 0 ~(=0) ~\Leftrightarrow~ \sum_{t=1}^{\infty}F_0(-r^{\fg}_t) <\infty ~(= \infty).    \end{equation}
By symmetry of the model, we also have 
\begin{equation}\label{eq:sum_infinite2}
    \Prob_0[\bar{a}=\fb] > 0 ~(=0) ~\Leftrightarrow~ \sum_{t=1}^{\infty}\big(1-F_0(r^{\fg}_t)\big) <\infty ~(= \infty). 
\end{equation}
In sum, asymptotic learning of informativeness holds if and only if conditioned on the source being uninformative, the probability of generating extreme signals decreases slowly enough so that both sums in \eqref{eq:sum_infinite1} and \eqref{eq:sum_infinite2} are infinite. As we discuss below, for uninformative signals with fatter tails, these sums diverge, and for signals with thinner tails, at least one of these sums converges. 

\subsection{Proof Sketch of Theorem \ref{thm:mainthm}} 
We conclude this section by providing a sketch of the proof of Theorem \ref{thm:mainthm}. On the one hand, by part (i) of Lemma \ref{lem:informative}, the probability of generating extreme signals that match the payoff-relevant state decreases relatively slowly under informative signals. Hence, if the source is uninformative and generates signals with fatter tails, this probability declines even more slowly. Consequently, both the sums in \eqref{eq:sum_infinite1} and \eqref{eq:sum_infinite2} are infinite, which means that no immediate agreement is possible. Thus, by Proposition \ref{cor:immd_learn_uniform}, asymptotic learning of informativeness holds. On the other hand, by part (ii) of Lemma \ref{lem:informative}, the probability of generating extreme signals that mismatch the state decreases relatively fast under informative signals. Hence, if the source is uninformative and generates signals with thinner tails, this probability decreases even more rapidly, at least for some type of extreme signals. Consequently, either the sum in \eqref{eq:sum_infinite1} or the sum in \eqref{eq:sum_infinite2} (or both) is finite, which means that immediate agreement on some action is possible. Hence, by Proposition \ref{cor:immd_learn_uniform}, asymptotic learning of informativeness fails.

\section{Discussion} \label{sec:discussion}
In this section, we discuss how relaxing the uniform-prior assumption affects learning about informativeness and the payoff-relevant state. We also discuss the connection between our model and misspecified learning.

\subsection{Non-Uniform Prior}
So far, we have assumed that all agents hold a uniform prior on the payoff-relevant state. Under this assumption, agents always act as if their private signals are informative (Lemma \ref{claim:opt_action}) and, more importantly, they always remain responsive to their signals. With a non-uniform prior, however, this responsiveness is no longer guaranteed.

To see this, suppose the prior strongly favors the good state. An agent who receives an extremely negative signal may interpret it as evidence that the source is uninformative if, given her prior, the signal is much more likely under an uninformative source than an informative one. She may then become sufficiently confident that the source is uninformative to disregard the signal and choose the action favored by her prior. In other words, a non-uniform prior creates a ``default'' action that the agent would choose if there is enough information uncertainty.\footnote{A similar force can arise when there is a third action that is strictly optimal when beliefs about $\theta$ are close to uniform. Under a uniform prior, this action can also serve as a default once agents become sufficiently confident that the source is uninformative.} Intuitively, the resulting nonresponsiveness hinders asymptotic learning of informativeness, since actions chosen regardless of private signals reveal no new information to the observer.

The next proposition formalizes this intuition and shows that asymptotic learning of informativeness fails under a non-uniform prior. 

\begin{proposition} \label{prop:nonuniform_fails}
Asymptotic learning of informativeness fails whenever the prior on the payoff-relevant state $\theta$ is non-uniform. 
\end{proposition}

Proposition \ref{prop:nonuniform_fails} shows that, for every fixed non-uniform prior, the outside observer remains uncertain about the source’s informativeness with positive probability, even in the limit. Nonetheless, this does not prevent him from approximately learning $\omega$. For example, it can be shown that for any fixed desired level of accuracy, the observer achieves that accuracy with high probability at a sufficiently large finite time,  provided that the prior is close enough to uniform and the uninformative signals have fatter tails.\footnote{Intuitively, since the observer eventually learns at the uniform prior, his belief can already be close to the truth with high probability at some finite time. Holding that time fixed, moving the prior slightly away from 1/2 changes the probabilities of action histories and the corresponding beliefs only slightly. Thus the observer remains highly accurate at that time, even though almost-sure learning in the limit fails.} Thus, the failure of asymptotic learning of informativeness need not imply that society remains substantially uncertain about the source's informativeness.

\subsection{Other Notions of Learning}
We now turn to learning about the payoff-relevant state under a non-uniform prior. In the standard model without information uncertainty, unbounded signals also ensure correct herding under a non-uniform prior. However, with information uncertainty, correct herding can fail even when the source is informative. As Proposition \ref{prop:nonuniform_fails} shows, when the prior is not uniform, the 
failure to learn that the source is informative may lead agents to disregard their signals and choose the action favored by their prior, which can be incorrect. Nevertheless, an argument analogous to the proof of Proposition \ref{prop:standardlearning} shows that, conditional on the source being informative, complete learning remains equivalent to asymptotic learning of informativeness. 

On the other hand, when the source is uninformative, it is straightforward to see that asymptotic learning of informativeness still implies the convergence of the social belief, now to the non-uniform prior. However, the converse implication in Proposition \ref{prop:standardlearning2} no longer holds in general. As in Proposition \ref{prop:nonuniform_fails}, a sufficiently strong prior may lead agents to choose its favored action regardless of their signals. Their action history then reveals no information about either the payoff-relevant state or the source’s informativeness. Thus, the social belief can remain at the prior even though society does not learn that the source is uninformative.



\subsection{Connections to Misspecified Learning}
We note that the main behavioral implication of the uniform-prior assumption, namely, that agents always behave as if their signals are informative, can be interpreted as a form of misspecified learning.\footnote{We thank an anonymous referee for pointing out this connection.} In particular, the same action process can be generated by agents who neglect the possibility that the signal source is uninformative and always perceive it as informative. Under this interpretation, agents’ perceived model is correctly specified when $\omega=1$ and misspecified when $\omega=0$.



To illustrate our main result under this interpretation, suppose $\omega=0$, so that agents misperceive uninformative signals as informative. Theorem \ref{thm:mainthm} shows that asymptotic learning of informativeness holds when the resulting perceived log-likelihood ratios have fatter tails than those generated by an informative source. Intuitively, agents then misperceive their signals\textemdash which are in fact uninformative\textemdash as more informative than the informative signals. Consequently, they overrespond to their own signals and overturn the herd frequently enough that their actions never converge. Such behavior then allows the outside observer to learn that the source is uninformative.

The idea that misspecification can make agents' actions more revealing is reminiscent of the mechanism in \cite{arieli2023hazards}. In their model, agents slightly underestimate the quality of others’ signals and consequently rely more heavily on their own signals. Their actions thus reveal more private information, generating a positive informational externality that improves the efficiency of learning. In our model, the misperception that signals are more informative (as reflected in fatter tails) instead leads agents to act against the herd infinitely often, which in turn allows the outside observer to identify an uninformative source.

\section{Conclusion}
In this paper, we study the sequential social learning problem in the presence of a potentially uninformative source. We show that achieving asymptotic learning of informativeness\textemdash where an outside observer eventually discerns whether the source is informative or uninformative\textemdash is not guaranteed, and it crucially depends on the relationship between the conditional distributions of the private signals. In particular, it hinges on the relative tail distribution of signals: when uninformative signals have fatter tails than their informative counterparts, the observer eventually learns the source's informativeness; when they have thinner tails, such learning fails. We also characterize the conditions for asymptotic learning of informativeness in the canonical case of Gaussian private signals, where the relative tail thickness is incomparable.


We make the symmetry assumption on the informative distributions $F_\fg$ and $F_\fb$ for expositional simplicity. Since our notion of relative tail thickness does not require this symmetry, our main result can be easily extended to the non-symmetric case. A more substantial assumption is that, under an uninformative source, agents are indifferent between the two actions. We see this as one way to ensure that agents continue to respond to their private signals, which is necessary for the observer to learn the source's informativeness asymptotically. As discussed in Section \ref{sec:discussion}, the behavior implied by the uniform-prior assumption in our model can also be generated by misspecified agents who neglect the possibility that the source is uninformative.

We have assumed throughout that informative and uninformative signals share a common support, so that agents cannot immediately recognize whether their information is useful. We believe some of our results also extend beyond this setting. For example, suppose the informative signals are the same as in Example \ref{ex:gaussian} while the uninformative signals are drawn uniformly from $[-\eps, \eps]$ for some small $\eps>0$. In this case, the distribution of private LLRs induced by these uninformative signals has bounded support, providing an extreme case of thinner tails. The same mechanism behind Theorem \ref{thm:mainthm} therefore should imply that asymptotic learning of informativeness fails. Interestingly, when the source is informative, each agent recognizes this with high probability for small $\eps$, since any signal outside $[-\eps,\eps]$ rules out the uninformative source; yet this individual recognition need not be revealed through actions.  In an alternative setting where a signal reveals that the source is uninformative, an agent with a uniform prior becomes indifferent between the two actions, so the observer’s ability to learn from her behavior may depend on the tie-breaking rule.




A limitation of our results is that they apply only asymptotically. Our numerical simulations suggest that, for uninformative Gaussian signals with fatter tails, the asymptotics can already kick in relatively early in the process. It would therefore be interesting to characterize how quickly an outside observer learns the source's informativeness. A natural extension, which we leave for future work, is to explore varying degrees of informativeness beyond the extreme cases considered here. For example, one could ask whether differences in the frequency of action switches can help an observer distinguish between weakly and strongly informative sources. An interesting technical question is whether our results extend to bounded signals. We conjecture that the characterization in Proposition \ref{thm:main} continues to hold in this setting, since agents in our model still eventually reach an action consensus under an informative source. However, our tail conditions do not directly apply to bounded private signals, so characterizing learning about informativeness in this setting would require further analysis. 


\appendix

\section{Omitted Proofs From the Main Text}
\numberwithin{equation}{section}

\setcounter{equation}{0}

\begin{proof}[Proof of Lemma \ref{claim:opt_action}]
Recall that in equilibrium, each agent $t$ chooses the action that is most likely to match the state $\theta$ conditioned on the available information $(H_t, s_t)$. By Bayes' rule, the relative likelihood between the good state and the bad state for agent $t$ is 
\begin{equation*}
\frac{\Prob[\theta = \fg|H_t, s_t]}{\Prob[\theta = \fb|H_t, s_t]} = \frac{\sum_{\Tilde{\omega} \in \{0, 1\}}\Prob_{\Tilde{\omega}}[\theta = \fg|H_t, s_t] \cdot \Prob[\omega=\Tilde{\omega}|H_t, s_t]}{\sum_{\Tilde{\omega} \in \{0, 1\}}\Prob_{\Tilde{\omega}}[\theta = \fb|H_t, s_t] \cdot \Prob[\omega=\Tilde{\omega}|H_t, s_t]}.
\end{equation*}
Note that $\Prob_0[\theta= \fg|H_t, s_t] = \Prob_0[\theta = \fg]$ and $\Prob_0[\theta= \fb|H_t, s_t] =\Prob_0[\theta = \fb]$ as conditioned on $\omega=0$, neither the history $H_t$ nor the signal $s_t$ contains any information about the payoff-relevant state $\theta$. Since the states $\omega$ and $\theta$ are independent of each other and the prior on $\theta$ is uniform, $\Prob_0[\theta= \fg|H_t, s_t] = \Prob_0[\theta= \fb|H_t, s_t] =1/2$. Thus, it follows from the above equation that 
$$\frac{\Prob[\theta = \fg|H_t, s_t]}{\Prob[\theta = \fb|H_t, s_t]} \geq 1 \Leftrightarrow \frac{\Prob_1[\theta = \fg|H_t, s_t]}{\Prob_1[\theta = \fb|H_t, s_t]} \geq 1.$$
That is, in equilibrium, each agent chooses the most likely action conditioned on the available information and the source being informative. 
\end{proof}

\begin{proof}[Proof of Lemma \ref{claim:overturning}]
For any $t \geq 1$, one has that 
    \begin{align*}
        \pi_{t+1} = \E_1[\mathbbm{1}(\theta=\fg)|H_{t+1}] =\E_1[\E_1[\mathbbm{1}(\theta=\fg)|H_t, s_t]|H_{t+1}],
    \end{align*}
    where the second equality follows from the law of the iterated expectation. Thus, if $a_t = \fg$, by Lemma \ref{claim:opt_action},  $\Prob_1[\theta=\fg|H_t, s_t] \geq \Prob_1[\theta=\fb|H_t, s_t]$. It follows from the above equation that $\pi_{t+1} \geq 1-\pi_{t+1}$, which implies that $\pi_{t+1} \geq 1/2$.  The case where $a_t = \fb$ implies that $\pi_{t+1} \leq 1/2$ follows from a symmetric argument.  
\end{proof}

The following simple claim will be useful in proving Proposition \ref{thm:main}. It employs an idea similar to the \emph{no introspection principle} in \cite{sorensen1996rational}. Recall that $q_t = \Prob[\omega =1|H_t]$ is the belief that an outside observer assigns to the source being informative based on the history of actions from time 1 to $t-1$ and $\gamma = \Prob[\omega=1] \in (0,1)$ is the prior that the source is informative.
\begin{claim} \label{claim:ineq}
For any $a \in (0, 1/2)$ and any $b\in (1/2, 1)$, 
   \begin{align*}
    \Prob_0[q_t = \Tilde{q}] &\leq \frac{1-a}{a} \frac{\gamma}{1-\gamma} \cdot 
 \Prob_1[q_t = \Tilde{q}], \text{~for all~} \Tilde{q} \in [a, 1/2];\\
       \Prob_0[q_t = \Tilde{q}] &\geq \frac{1-b}{b} \frac{\gamma}{1-\gamma}\cdot \Prob_1[q_t = \Tilde{q}], \text{~for all~} \Tilde{q} \in [1/2, b].
   \end{align*} 
\end{claim}
\begin{proof}
Fix a prior $\gamma\in (0,1)$. For any $\tilde{q} \in (0,1)$, let $\Tilde{H}_t$ be a history of actions such that the associated belief $q_t$ is equal to $\Tilde{q}$. By the law of total expectation, 
\begin{align*}
\Prob[\omega=1|q_t = \Tilde{q}] 
& = \E[\E[\mathbbm{1}(\omega=1)| \Tilde{H}_t] | q_t = \Tilde{q}]  = \E[\Tilde{q} |q_t = \Tilde{q}] = \Tilde{q}.
\end{align*}
It follows from Bayes' rule that 
    \begin{align*}
        \frac{\Prob_0[q_t = \Tilde{q}]}{\Prob_1[q_t = \Tilde{q}]} =  \frac{\Prob[\omega=0|q_t = \Tilde{q}]}{\Prob[\omega=1|q_t = \Tilde{q}]} \cdot \frac{\Prob[\omega=1]}{\Prob[\omega=0]} = \frac{1-\Tilde{q}}{\Tilde{q}} \frac{\gamma}{1-\gamma}.
    \end{align*}
Since for any $a \in (0, 1/2)$ and any $\Tilde{q} \in [a, 1/2]$, $1 \leq \frac{1-\Tilde{q}}{\Tilde{q}} \leq \frac{1-a}{a}$, it thus follows from the above equation that $$\Prob_0[q_t = \Tilde{q}] = \frac{1-\Tilde{q}}{\Tilde{q}} \frac{\gamma}{1-\gamma} \cdot \Prob_1[q_t = \Tilde{q}] \leq \frac{1-a}{a} \frac{\gamma}{1-\gamma} \cdot  \Prob_1[q_t = \Tilde{q}].$$ The second inequality follows from an identical argument.
\end{proof}

In the following proofs of Propositions \ref{thm:main} and \ref{prop:standardlearning}, we use extensively the idea of an outside observer, say observer $x$, who observes everyone's actions. The information available to him at time $t$ is thus $H_t = (a_1, \ldots, a_{t-1})$ and at time infinity is $H_{\infty} = \cup_t H_t$. He gets a utility of one if his guess about $\omega$ is correct and zero otherwise. Furthermore, we denote by $q_{\infty} := \Prob[\omega=1|H_{\infty}]$ the belief that he has at time infinity about the event $\omega=1$. Similarly, we denote by $p_{\infty} :=\Prob[\theta=\fg|H_{\infty}]$ the belief that he has at time infinity about the event $\theta= \fg$.

\begin{proof}[Proof of Proposition \ref{thm:main}]
\emph{(If direction)} 
Suppose $\Prob_{0}[S = \infty]= 1$. Let  $a^x_t$ be the guess that the outside observer $x$ would make to maximize his probability of guessing $\omega$ correctly at time $t$. Fix a large positive integer $k \in \mathbbm{N}$. Let $A_t(k)$ denote the event that there have been at least $k$ action switches before time $t$ and denote its complementary event by $A^c_t(k)$. Consider the following strategy $\Tilde{a}^x_\infty(k)$ for $x$ at time infinity: $\Tilde{a}^x_\infty(k) = 0$ if $A_\infty(k)$ occurs and $\Tilde{a}^x_\infty(k) =  1$ otherwise. 
The probability of $x$ guessing it correctly under this strategy is
\begin{align} \label{eq:x_strategy1}
    \Prob[\Tilde{a}^x_\infty(k) = \omega]  
     = \Prob_0[A_{\infty}(k)] \cdot \Prob[\omega=0] + \Prob_1[A^c_{\infty}(k)] \cdot \Prob[\omega=1].
\end{align}
Since $\Prob_{1, \theta}[a_t \to \theta]=1$, we have $\Prob_1[S < \infty] = 1$, and hence
\begin{equation} \label{eq:smith}
   \lim_{k \to \infty}\Prob_1[A^c_{\infty}(k)] = 1.
\end{equation} 
By assumption, $\Prob_{0}[S = \infty]= 1$, and thus $\Prob_{0}[A_\infty(k)] =1$ for all $k$. Together, \eqref{eq:x_strategy1} and \eqref{eq:smith} imply that 
 $$ \lim_{k \to \infty}\Prob[\Tilde{a}^x_\infty(k) =\omega] =1.$$  

Meanwhile, the optimal strategy for $x$ at time infinity is the following:  $a^x_\infty = 1$ if $q_\infty \geq 1/2$ and $a^x_\infty =0$ otherwise. 
Since this strategy does at least as well as $\tilde a_\infty^x(k)$ for every $k$, taking the limit as $k \to \infty$ gives
\begin{align} \label{eq:opt_outsider}
  1 =  \Prob[a^x_{\infty} =\omega] = \Prob_1[q_\infty \geq 1/2]\cdot \Prob[\omega=1] + \Prob_0[q_{\infty} < 1/2] \cdot \Prob[\omega=0].
\end{align}
From \eqref{eq:opt_outsider},  $\Prob_0[q_{\infty} < 1/2] = \Prob_1[q_{\infty} \geq 1/2] =1$. To establish asymptotic learning of informativeness, by definition, we need to show $\Prob_0[q_{\infty} =0] =1$ and $\Prob_1[q_{\infty} =1] =1$. To this end, note that $\Prob_0[q_{\infty} < 1/2] =1$ implies that
$\Prob_0[q_{\infty} \geq  1/2] = 0$. Thus, by Claim \ref{claim:ineq}, it implies that for any $b \in (1/2, 1)$ and all $\Tilde{q} \in [1/2, b]$,
$\Prob_1[q_{\infty} = \Tilde{q}] = 0$. Consequently, it follows from  $\Prob_1[q_{\infty} \geq 1/2] =1$ that $\Prob_{1}[q_\infty =1]=1$. The case that $\Prob_0[q_{\infty} =0] =1$ follows from an identical argument. 


\emph{(Only-if direction)} Suppose by contraposition that $\Prob_{0}[S < \infty]>0$. Again, since $\Prob_{1, \theta}[a_t \to \theta] =1$, we have $\Prob_{1}[S < \infty] =1$. Thus, there exists a history of actions at time infinity $\Tilde{H}_\infty$  that is possible under both probability measures $\Prob_0$ and $\Prob_1$: $\Prob_0[\Tilde{H}_\infty] >0$ and $\Prob_1[\Tilde{H}_\infty] >0$. By Bayes' rule,  $\Prob[\omega=1|\Tilde{H}_{\infty}] <1$ and $\Prob[\omega=0|\Tilde{H}_{\infty}]<1$. Assume without loss of generality that under  $\Tilde{H}_\infty$, the corresponding belief $\Tilde{q}_{\infty} \geq 1/2$, and so the observer $x$ would guess $a^x_{\infty} = 1$. Hence, the probability of $x$ guessing correctly about $\omega$ is strictly less than one: 
\begin{align} 
    \Prob[a^x_{\infty} =\omega] 
    & = \Prob[a^x_{\infty} =\omega, \Tilde{H}_{\infty}] +  \Prob[a^x_{\infty} =\omega, \Tilde{H}^c_{\infty}] \nonumber \\
    & = \Prob[\omega= 1|\Tilde{H}_{\infty}] \cdot \Prob[\Tilde{H}_{\infty}] + \Prob[a^x_{\infty}= \omega,\Tilde{H}^c_{\infty}]  \nonumber < \Prob[\Tilde{H}_{\infty}] + \Prob[\Tilde{H}^c_{\infty}]= 1. \nonumber
\end{align}
This contradicts $\Prob[a^x_\infty = \omega] = 1$, which is required for asymptotic learning of informativeness to hold. 
\end{proof}

The following decomposition will be useful in proving Propositions \ref{prop:standardlearning} and \ref{prop:standardlearning2}, and later in Proposition \ref{prop:nonuniform_fails}.
Recall that $q_{\infty}$  and $p_{\infty}$ are the beliefs of the observer $x$ assigned to the events $\omega =1$ and $\theta = \fg$ at time infinity, respectively. By the law of total probability,
\begin{align} 
    p_{\infty} 
     & = \Prob_1[\theta= \fg|H_{\infty}] \cdot q_{\infty} + \Prob_0[\theta= \fg|H_{\infty}] \cdot (1- q_{\infty}) \label{eq:standardLearning} \\
     & =  \pi_{\infty} \cdot q_{\infty} + \frac{1}{2}\cdot (1- q_{\infty}).  \label{eq:standardLearning2}
\end{align}
Here, $\pi_\infty = \Prob_1[\theta= \fg|H_{\infty}]$, and the second equality follows from the fact that, conditioned on $\omega=0$, the action history contains no information about $\theta$ and the prior on $\theta$ is uniform.  

\begin{proof}[Proof of Proposition \ref{prop:standardlearning}]
We first note that because agents always act as if signals are informative, and signals are unbounded, the standard learning result \citep{smith2000pathological} implies that the agent's public belief $\pi_t = \Prob_1[\theta =\fg|H_t]$ converges to the truth when $\omega=1$. That is,  (i) conditioned on $\omega=1$ and $\theta= \fg$, $\pi_\infty=1$ almost surely and (ii) conditioned on $\omega=1$ and $\theta=\fb$, $\pi_\infty=0$ almost surely. 

Suppose $\omega=1$ and asymptotic learning of informativeness holds. By assumption, $q_\infty=1$ almost surely. Substituting into \eqref{eq:standardLearning2}, we get $p_\infty = \pi_\infty$. Since $\pi_t$ converges to the truth, $p_t$ also converges to the truth. Therefore, complete learning holds. Conversely, suppose complete learning holds conditional on $\omega =1$. If $\theta = \fg$, then $p_\infty =  \pi_\infty= 1$ almost surely. It follows from \eqref{eq:standardLearning2} that 
$$1 = q_\infty + \frac{1}{2}(1-q_\infty)$$ 
which implies that $q_\infty =1$. If $\theta = \fb$, then $p_\infty =  \pi_\infty= 0$ almost surely. Again by \eqref{eq:standardLearning2}, 
\[
0 = \frac{1}{2}(1-q_\infty) 
\]
which also implies $q_\infty = 1$. Therefore, conditional on $\omega = 1$, asymptotic learning of informativeness holds. 
\end{proof}

\begin{proof}[Proof of Proposition \ref{prop:standardlearning2}]
The only-if direction is straightforward. Suppose $\omega = 0$ and asymptotic learning of informativeness holds. By definition, $q_\infty =0$ almost surely. Substituting this into \eqref{eq:standardLearning2}, we get $p_\infty = 1/2$ almost surely, as required. 

We prove the if direction by contraposition. Suppose $\omega =0$ and asymptotic learning of informativeness fails. That is, the event $\{q_\infty > 0\}$ occurs with some positive probability.  
By Proposition \ref{thm:main}, agents eventually reach an action consensus with positive probability. As in the proof of Proposition \ref{thm:main}, we can choose a consensus history $\tilde{H}_\infty$  that occurs with positive probability under $\Prob_0$ and is also possible under $\Prob_1$. Therefore, the corresponding belief of an outside observer satisfies $\tilde{q}_\infty  \in (0,1)$.\footnote{Note that conditioned on $\omega=0$, $q_\infty$ has support $\subseteq [0, 1)$ as $\Prob[\omega=0|q_\infty = 1] =0$.} Moreover, as shown in the proof of Proposition \ref{prop:standardlearning}, the corresponding public belief of the agent $\tilde{\pi}_\infty$ takes values in $\{0, 1\}$, depending on the action to which $\tilde{H}_\infty$ converges. If $\tilde \pi_\infty = 1$, then by \eqref{eq:standardLearning2}, $\tilde p_\infty = \frac{1}{2}(1+\tilde q_\infty) > 1/2$. If $\tilde \pi_\infty = 0$, then by \eqref{eq:standardLearning2}, $\tilde p_\infty = \frac{1}{2}(1-\tilde q_\infty) < 1/2$. In either case, $p_\infty \neq 1/2$. Hence, we conclude that the event $\{p_\infty \neq 1/2\}$ occurs with positive probability, establishing the contraposition.
\end{proof}

\begin{proof}[Proof of Lemma \ref{lem:informative}] 
The proof idea is similar to the proof of Lemma 10 in \cite*{arieli2023hazards}. 
Recall that we use $\Prob_{1, \fg}$ to denote the conditional probability distribution given $\omega=1$ and $\theta=\fg$ and we use $\Prob_{1, \fg, \pi}$ to denote the same conditional probability distribution while emphasizing the prior value. Since conditioned on $\omega=1$, correct herding holds, and this means that $\Prob_{1, \fg}[a_t \to \fg] = 1$. As a consequence, part (i) follows directly from the fact that the events $\{\bar{a} = \fb\}$ and $\{a_t \to \fg\}$ are disjoint, and thus $\Prob_{1,\fg}[\bar{a} = \fb] = 0$. 

For part (ii), let $\tau < \infty$ denote the last random time at which the agent chooses the wrong action $\fb$. It is well-defined as correct herding holds. Hence, $1 = \Prob_{1,\fg}[a_t \to \fg] = \sum_{k = 1}^{\infty} \Prob_{1,\fg}[\tau=k]$. By the overturning principle (Lemma \ref{claim:overturning}), $a_{\tau} = \fb$ implies that $\pi_{\tau +1} \leq 1/2$. As a consequence, 
\begin{align*}
1=\sum_{k = 1}^{\infty} \Prob_{1,\fg}[\tau=k] & = \sum_{k = 1}^{\infty} \Prob_{1,\fg}[a_{k+1} = a_{k+2} = \ldots = \fg, \pi_{k+1} \leq 1/2]\\
      & = \sum_{k = 1}^{\infty} \E_{1,\fg}\big[\Prob_{1,\fg}[a_{k+1} =  a_{k+2} = \ldots = \fg, \pi_{k+1}\leq 1/2~|~\pi_{k+1}]\big]  \\
     & = \sum_{k = 1}^{\infty} \E_{1,\fg}\big[ \mathbbm{1}(\pi_{k+1}\leq 1/2) \cdot \Prob_{1,\fg}[a_{k+1} =  a_{k+2} = \ldots = \fg~|~ \pi_{k+1}]\big] \\
     & = \sum_{k = 1}^{\infty} \E_{1,\fg}\big[ \mathbbm{1}(\pi_{k+1}\leq 1/2) \cdot \Prob_{1,\fg, \pi_{k+1}}[\bar{a} = \fg]\big],
\end{align*}
where the second equality follows from the law of total expectation, and the last equality follows from the stationarity property (Lemma \ref{claim:stationary}). Suppose that for all priors $\pi \in (0,1)$, $\Prob_{1,\fg, \pi}[\bar{a} = \fg] = 0$. This implies that the above equation equals zero, a contradiction. 
\end{proof}

\subsection{Proof of Theorem \ref{thm:mainthm}}
In this section, we prove Proposition \ref{cor:immd_learn_uniform} and Theorem \ref{thm:mainthm}. To prove Proposition  \ref{cor:immd_learn_uniform}, we will first prove the following proposition (Proposition \ref{prop:learning_immediate_herd}). Together with Proposition \ref{thm:main}, they jointly imply Proposition \ref{cor:immd_learn_uniform}. The proof of Theorem \ref{thm:mainthm} is presented at the end of this section. We write $\Prob_{0, \pi}$ to denote the conditional probability distribution given $\omega=0$ while highlighting the value of the prior $\pi$ on $\theta$.\footnote{We continue to omit the prior $\pi$ when it is uniform.} 

\begin{proposition} \label{prop:learning_immediate_herd}
The following are equivalent. 
\begin{compactenum}[(i)]
\item For any action $a \in \{\fg, \fb\}$, $\Prob_{0, \pi}[\bar{a}= a]=0$ for all priors $\pi \in (0,1)$. 
\item $\Prob_0[S = \infty] = 1$.
\item For any action $a \in \{\fg, \fb\}$, $\Prob_{0, \pi}[\bar{a} = a] =0$ for some prior $\pi\in (0,1)$. 
\end{compactenum}
\end{proposition}

To prove this proposition, we first establish some preliminary results on the process of the agents' public log-likelihood ratios conditioned on the event of immediate agreement. These results lead to Lemma \ref{prop:pos_herd}, a crucial part in establishing the equivalence between no immediate agreement and perpetual disagreement conditioned on an uninformative source. We present the proof of Proposition \ref{prop:learning_immediate_herd} towards the end of this section.

\subsubsection*{Preliminaries}
Recall that conditioned on $\{\bar{a} = \fg\}$, the process of the agent's public log-likelihood ratio $r_t$ evolves deterministically according to $\eqref{eq:highaction}$, which we denote by $r^{\fg}_t$. Let the corresponding updating function be    
\[
\phi(x) := x+D_{\fg}(x).
\]
That is, $r^\fg_{t+1} = \phi(r^{\fg}_t)$ for all $t\geq 1$. Since the entire sequence $(r^{\fg}_t)$ is determined once its initial value $r_1^\fg$ is specified, we denote the value of $r^{\fg}_t$ with an initial value $r_1^\fg = r$ by $r^{\fg}_t(r)$. We can thus write $r^{\fg}_t(r) = \phi^{t-1}(r)$ for all $t \geq 1$, where $\phi^t$ is its $t$-th composition and $\phi^0(r)= r$. 

We remind the readers of two standard properties of the sequence $(r^{\fg}_t)$, as summarized in the following lemma. The first part of this lemma states that $(r^{\fg}_t)$ tends to infinity as $t$ tends to infinity, and the second part shows that it takes only some bounded time for the sequence $(r^\fg_t)$ to reach any positive value. 


\begin{lemma}[The Long-Run and Short-Run Behaviors of $r^\fg_t$]\label{lemma:rg}\hfill
\begin{compactenum}[(i)]
\item $\lim_{t \to \infty} r^{\fg}_t = \infty $. 
    \item For any $\bar{r} \geq 0$, there exists $t_0$ such that $r^{\fg}_{t_0}(r) \geq \bar{r}$ for all $r \geq 0$.
\end{compactenum}
\end{lemma}
\begin{proof}
  See Lemma 6 and Lemma 12 in \cite{rosenberg2019efficiency}. 
  \end{proof}
  
Although the sequence $(r^{\fg}_t)$ increases to infinity, the updating function $\phi$ need not be monotone.\footnote{In the case of binary states and actions, \cite{herrera2012necessary} show that the property of increasing hazard ratio is equivalent to the monotonicity of this updating function. See \cite*{smith2021informational} for a general treatment.} The next lemma shows that, under some mild technical assumptions on the left tail of $F_{\fb}$, the function $\phi(x)$ eventually increases monotonically. 


\begin{lemma}[Eventual Monotonicity]\label{lem:eventual_mon}
Suppose that $F_{\fb}$ has a differentiable left tail and its probability density function $f_{\fb}$ satisfies the condition that $\limsup_{x\to \infty} f_\fb(-x) < 1$. Then, $\phi(x): = x + D_\fg(x)$ increases monotonically for all $x$ large enough.
\end{lemma}

\begin{proof}
By assumption, we can find a constant $\rho <1$ such that for all $x$ large enough, $f_{\fb}(-x) \leq \rho$. By definition, $D_g(x) = \log \frac{1-F_{\fg}(-x)}{1-F_{\fb}(-x)}$. Taking the derivative of $D_{\fg}$, 
\[
D_{\fg}'(x) = \frac{f_{\fg}(-x)}{1-F_{\fg}(-x)} - \frac{f_{\fb}(-x)}{1-F_{\fb}(-x)}. 
\]
Observe that the log-likelihood ratio of the agent's private log-likelihood ratio $\ell_t$ is the log-likelihood ratio itself (see, e.g., \cite{chamley2004rational}):
\begin{align*}
   \log \frac{dF_{\fg}}{dF_{\fb}}(x) = x.
\end{align*}
It follows that
\begin{align*}
    -D'_\fg(x) = f_{\fb}(-x) \Big( \frac{1}{1-F_{\fb}(-x)} - \frac{\ee^{-x}}{1-F_{\fg}(-x)} \Big)  \leq   \frac{f_{\fb}(-x)}{1-F_{\fb}(-x)}.
\end{align*}
Fix some $\eps>0$ small enough so that $(1+\eps) \rho \leq 1$. It follows from the above inequality that there exists some $x$ large enough such that $-D'_\fg(x) \leq (1+\eps) f_{\fb}(-x)$. Furthermore, for all $x' \geq x$,
\begin{align*}
         D_{\fg}(x) &=  D_{\fg}(x') - \int_x^{x'} D'_{\fg}(y) dy\\
         & \leq  D_{\fg}(x') + (1+\eps) \int_x^{x'} f_{\fb}(-x) dx\\
         & = D_{\fg}(x') - (1+\eps) (F_{\fb}(-x')- F_{\fb}(-x)).
    \end{align*}
Rearranging the above equation, 
    \begin{align*}
        D_{\fg}(x) - D_{\fg}(x') &\leq (1+\eps)(F_{\fb}(-x) - F_{\fb}(-x'))\\
        & \leq (1+\eps)\rho  (x'-x),
    \end{align*}
where the last inequality follows from the fact that $f_{\fb}(-x) \leq \rho < 1$. Since $(1+\eps)\rho \leq 1$, the above inequality implies that there exists some $x$ large enough such that $D_{\fg}(x)  + x \leq D_{\fg}(x')+ x'$ for all $ x' \geq x$. That is, $\phi(x)$ eventually increases monotonically. 
\end{proof}

Given these lemmas, we are ready to prove the following result. It shows that conditioned on an uninformative source, the possibility of immediate agreement is independent of the prior belief.  

\begin{lemma}\label{prop:pos_herd}
For any action $a\in \{\fg, \fb\}$, the following statements are equivalent:
    \begin{compactenum}[(i)]
        \item $ \Prob_{0, \pi}[\Bar{a} = a] >0$, for some prior $\pi \in (0,1)$;
        \item $ \Prob_{0, \pi}[\Bar{a} = a] >0$, for all priors $\pi \in (0,1)$.
    \end{compactenum}
\end{lemma} 
\begin{proof}
    The second implication, namely, $(ii) \Rightarrow (i)$ is immediate. We will show the first implication, $(i) \Rightarrow (ii)$. Fix some prior $\Tilde{\pi} \in (0,1)$ such that $\Prob_{0,\Tilde{\pi}}[\Bar{a}=\fg] >0$ and let $\Tilde{r} = \log \frac{\Tilde{\pi}}{1-\Tilde{\pi}}$. Since $r^{\fg}_t(\Tilde{r})$ is a deterministic process, the event $\{\bar{a}= \fg\}$ initiated at the prior $\pi_1=\Tilde{\pi}$ is equivalent to the event $\{\ell_t \geq -r^{\fg}_t(\Tilde{r}), \forall t \geq 1\}$. Conditioned on $\omega=0$, since signals are i.i.d., so are the agents' private log-likelihood ratios. Thus, we have
\begin{align}\label{eq:product_prob}
    \Prob_{0,\Tilde{\pi}}[\Bar{a} = \fg]  & = \prod_{t=1}^\infty (1-F _{0}(- r^{\fg}_t(\Tilde{r})).
    \end{align}
    As a consequence, $\Prob_{0,\Tilde{\pi}}[\Bar{a}=\fg] >0$ if and only if there exists $M<\infty$ such that 
    \begin{equation*}
       - \sum_{t=1}^{\infty} \log \big(1- F_{0}(- r^{\fg}_t(\Tilde{r})) \big) <M.
    \end{equation*}
For two sequences $(a_t)$ and $(b_t)$, we write $a_t \approx b_t$ if $\lim_{t \to \infty} (a_t/b_t) =1$. Since $r^{\fg}_t(\Tilde{r}) \to \infty$ (this follows from part (i) of Lemma \ref{lemma:rg}), $\log \big(1- F_{0}(- r^{\fg}_t(\Tilde{r})) \big) \approx -F_{0}( - r^{\fg}_t(\Tilde{r}))$. Thus, the above sum is  finite  if and only if 
    \begin{equation} \label{eq:finitesum}
      \sum_{t=1}^{\infty}  F_{0}(- r^{\fg}_t(\Tilde{r})) <M.  
 \end{equation}
By the overturning principle (Lemma \ref{claim:overturning}), it suffices to show that \eqref{eq:finitesum}  implies that 
\[
 \sum_{t=1}^{\infty}  F_{0}(- r^{\fg}_t(r)) <M, \quad \text{for any~} r\geq 0. 
\]
By the eventual monotonicity of $\phi$ (Lemma \ref{lem:eventual_mon}) and the fact that $r^\fg_t(\Tilde{r}) \to \infty$, we can find a large enough $\bar{t}$ such that $r^{\fg}_{\bar{t}}(\Tilde{r}) := \bar{r} \geq 0$ and $\phi(r) \geq \phi(\bar{r})$ for all $r \geq \bar{r}$. By part (ii) of Lemma  \ref{lemma:rg}, there exists $t_0 \in \mathbbm{N}$ such that $r^{\fg}_{t_0}(r) \geq \bar{r}$ for all $r\geq 0$. Since above $\bar{r}$, $\phi$ is monotonically increasing, one has $\phi(r^{\fg}_{t_0}(r)) \geq \phi(\bar{r})$ for any $r \geq 0$. Consequently, for all $\tau\geq 1$, $r^{\fg}_{\tau+t_0}(r) = \phi^{\tau}(r^{\fg}_{t_0}(r)) \geq \phi^{\tau}(\bar{r}) = r^{\fg}_{\tau+1}(\bar{r})$. Since $r^{\fg}_{\tau+1}(\bar{r}) = r^{\fg}_{\tau+1}(r^{\fg}_{\bar{t}}(\Tilde{r})) = r^{\fg}_{\tau+\bar{t}}(\Tilde{r})$, it follows that 
$$F_0(-r^{\fg}_{\tau+t_0}(r)) \leq F_0(-r^{\fg}_{\tau+\bar{t}}(\Tilde{r})).$$ 
Thus, it follows from \eqref{eq:finitesum} that for any $r\geq 0$,
$ \sum_{t=1}^{\infty} F_{0}(- r^{\fg}_t(r)) < \infty$,  as required. The case for action $\fb$ follows from a symmetric argument. 
\end{proof}

Now, we are ready to prove Proposition \ref{prop:learning_immediate_herd}. 
\begin{proof}[Proof of Proposition \ref{prop:learning_immediate_herd}]
We show that $(i) \Rightarrow (ii)$, $(ii) \Rightarrow (iii)$, and $(iii) \Rightarrow (i)$. To show the first implication, we prove the contrapositive statement. Suppose that $\Prob_{0}[S < \infty]>0$.  This implies that there exists a sequence of action realizations $(b_1, b_2, \ldots, b_{k-1}, b_k = \ldots = a)$ for some action $a \in \{\fb, \fg\}$ such that 
\[
\Prob_{0}[a_t = b_t, \forall t\geq 1] >0.
\]
By stationarity, there exists some $\pi' \in (0, 1)$ such that
$$\Prob_{0,\pi'}[\Bar{a} = a]>0,$$ 
which contradicts (i). 

To show the second implication, suppose towards a contradiction that there exists some action $a \in \{\fg, \fb\}$ such that $\Prob_{0,\pi}[\Bar{a} = a] >0$  for all priors $\pi \in (0, 1)$. In particular, it holds for the uniform prior. Since the event $\{\bar{a} = a\}$ is contained in the event $\{S < \infty\}$, 
    \[
    0 <\Prob_{0}[\Bar{a} = a] \leq \Prob_{0}[S < \infty].
    \]
This implies that $\Prob_0[S =\infty] <1$, a contradiction to (ii). 

Finally, we show the last implication by contraposition. Suppose that there exists some $a \in \{\fg, \fb\}$ such that $\Prob_{0, \pi}[\bar{a} = a] >0$ for some prior $\pi\in (0,1)$. By Lemma \ref{prop:pos_herd}, it also holds for all priors $\pi \in (0,1)$, which is a contradiction to (iii). This concludes the proof of Proposition \ref{prop:learning_immediate_herd}.
\end{proof}
\begin{proof}[Proof of Proposition \ref{cor:immd_learn_uniform}]
   By the equivalence between (ii) and (iii) in Proposition \ref{prop:learning_immediate_herd} and Proposition \ref{thm:main}, we have shown Proposition \ref{cor:immd_learn_uniform}.
\end{proof}

Given Proposition \ref{cor:immd_learn_uniform}, we are now ready to prove our main result.
\begin{proof}[Proof of Theorem \ref{thm:mainthm}]
By part (i) of Lemma \ref{lem:informative}, $\Prob_{1,\fb}[\bar{a} = \fg] = 0$ and $\Prob_{1,\fg}[\bar{a} = \fb] = 0$. Following a similar argument that led to \eqref{eq:product_prob}, one has
\begin{align*}
0 = \Prob_{1,\fb}[\bar{a} = \fg] = \prod_{t=1}^{\infty} (1 - F_{\fb}(-r^{\fg}_t)).
\end{align*}
Taking the logarithm on both sides, the above equation is equivalent to $-\sum_{t=1}^{\infty}\log (1 - F_{\fb}(-r^{\fg}_t)) = \infty$. Since $r^{\fg}_t \to \infty$,  $\log (1 - F_{\fb}(-r^{\fg}_t))  \approx - F_\fb(-r^{\fg}_t)$ and the previous sum is infinite if and only if 
\begin{equation} \label{eq:proof_mainthem1}
     \sum_{t=1}^{\infty}F_\fb(-r^{\fg}_t) = \infty.
\end{equation}
Similarly, we have $\Prob_{1,\fg}[\bar{a} = \fb] = 0$ if and only if $ \sum_{t=1}^{\infty}(1-F_\fg(-r^{\fb}_t)) = \infty$, where $r_t^{\fb}$ denotes the deterministic process of $r_t$ conditioned on the event $\{\bar{a} =\fb\}$. By symmetry, $r_t^{\fb} = -r_t^{\fg}$ for all $t \geq 1$. Hence, $\Prob_{1,\fg}[\bar{a} = \fb] = 0 $ if and only if
\begin{equation} \label{eq:proof_mainthem11}
\sum_{t=1}^{\infty}(1-F_\fg(r^{\fg}_t)) = \infty,
\end{equation}
 
Suppose that the uninformative signals have fatter tails than the informative signals. By definition, there exists $\eps > 0$ such that for all $x$ large enough,
$F_0(-x) \geq \eps \cdot F_{\fb}(-x)$ and $1 - F_0(x) \geq \eps \cdot (1 - F_{\fg}(x))$. It then follows from \eqref{eq:proof_mainthem1} and \eqref{eq:proof_mainthem11} that $$\sum_{t=1}^{\infty}F_0(-r^{\fg}_t) = \infty \quad \text{and} \quad \sum_{t=1}^{\infty}(1-F_0(r^{\fg}_t)) = \infty.$$ Using the same logic we used to deduce
\eqref{eq:proof_mainthem1} and \eqref{eq:proof_mainthem11}, having these two divergent sums is equivalent to $\Prob_0[\bar{a} = \fg] =0$ and $\Prob_0[\bar{a} = \fb] =0$. 
Thus, by Proposition \ref{cor:immd_learn_uniform}, asymptotic learning of informativeness holds. 

By part (ii) of Lemma \ref{lem:informative}, there exist $\pi, \pi' \in (0,1)$  such that $\Prob_{1,\fg, \pi}[\bar{a} = \fg] > 0$ and $\Prob_{1,\fb, \pi'}[\bar{a} = \fb] > 0$.  Let $r = \log \frac{\pi}{1-\pi}$ and $r' = \log \frac{\pi'}{1-\pi'}$. Following a similar argument that led to \eqref{eq:finitesum}, these are equivalent to 
\[
      \sum_{t=1}^{\infty} F_{\fg}(-r^{\fg}_t(r)) < \infty \quad \text{and} \quad  \sum_{t=1}^{\infty}(1 -F_{\fb}(r^{\fg}_t(r'))) < \infty.
 \]
 
Now, suppose that the uninformative signals have thinner tails than the informative signals. By definition, there exists  $\eps >0$ such that either (i) $F_0(-x) \leq (1/\eps) \cdot F_{\fg}(-x)$ for all $x$ large enough, or (ii) $1-F_0(x) \leq (1/\eps) \cdot (1-F_{\fb}(x))$ for all $x$ large enough. It then follows from the above inequalities that either (i) $ \sum_{t=1}^{\infty} F_{0}(-r^{\fg}_t(r)) < \infty$ or (ii) $\sum_{t=1}^{\infty}(1 -F_{0}(r^{\fg}_t(r'))) < \infty$. Equivalently, we have either (i) $\Prob_{0, \pi}[\bar{a} = \fg] >0$ for some $\pi \in (0,1)$ or (ii) $\Prob_{0, \pi'}[\bar{a} = \fb] >0$ for some $\pi' \in (0,1)$. By Lemma \ref{prop:pos_herd}, these also hold for all priors $\pi, \pi' \in (0,1)$, including the uniform prior. So we have either (i) $\Prob_{0}[\bar{a} = \fg] >0$ or (ii) $\Prob_{0}[\bar{a} = \fb] >0$.  Thus, by Proposition \ref{cor:immd_learn_uniform}, asymptotic learning of informativeness fails. 
\end{proof}

\section{Gaussian Private Signals} \label{sec:appendix_guassian}
In this section, we consider the canonical environment in which all signal distributions $\mu_\fg$, $\mu_\fb$, and $\mu_0$ are normal. Specifically, $\mu_\fg$ and $\mu_\fb$ share the same variance $\sigma^2$ and have means $+1$ and $-1$, respectively, while $\mu_0$ has mean $m_0 \in (-1, 1)$ and variance $\tau^2$. The agent's private log-likelihood ratio induced by a signal $s$ is therefore 

\begin{equation}\label{eq:gaussian_private}
\ell_t = \log \frac{d\mu_\fg}{d\mu_\fb}(s) = \frac{2}{\sigma^2} s.
\end{equation}
Since $\ell$ is proportional to $s$, the distributions $F_\theta$ and $F_0$ are also normal, with a variance of $4/\sigma^2$ and $4\tau^2/\sigma^4$, respectively.

We first consider the case where $\tau \neq \sigma$. Here, the relative tail thickness between the uninformative and informative signals is determined solely by their variances, as shown by the following lemma. 
 
\begin{lemma} \label{lem:variance_tails}
Suppose $F$ and $G$ are two Gaussian cumulative distribution functions with means $\mu_1$ and $\mu_2$ and variances $\sigma_1^2$ and $\sigma_2^2$, respectively. If $\sigma_1 > \sigma_2$, then $F$ has fatter tails than $G$. Meanwhile, $G$ has thinner tails than $F$.
\end{lemma}

\begin{proof}
Let $f$ and $g$ denote the probability density functions of $F$ and $G$. Their ratio evaluated at $x \in \mathbbm{R}$ is 
\begin{align*}
\frac{f(x)}{g(x)}  =   \frac{\sigma_2}{\sigma_1} \exp{\Big((\frac{1}{\sigma_2^2} - \frac{1}{\sigma_1^2})\frac{x^2}{2} +(\frac{\mu_1}{\sigma_1^2} - \frac{\mu_2}{\sigma_2^2})x +  \frac{1}{2}(\frac{\mu_2^2}{\sigma_2^2} - \frac{\mu_1^2}{\sigma_1^2})\Big)}.
\end{align*}
Suppose $\sigma_1> \sigma_2$. It follows from the above equation that $\lim_{x \to \infty} \frac{f(-x)}{g(-x)} = \infty$ and $\lim_{x \to \infty} \frac{f(x)}{g(x)} = \infty$. This clearly implies that $F(-x)\geq G(-x)$ and $1-F(x) \geq 1-G(x)$ for all $x$ large enough. Thus, by definition, $F$ has fatter tails than $G$, and conversely, $G$ has thinner tails than $F$.
\end{proof}

We now consider the case where $\tau = \sigma$. As discussed in the main text, this is a special case where the uninformative signals have neither fatter nor thinner tails. 

\begin{proof}[Proof of Proposition \ref{prop:gaussian}]
Throughout the proof, we write $f(x)\sim g(x)$ to mean that
$f(x)/g(x)\to1$ in the indicated limit. Recall that $r_t^\fg$ denotes the public LLR along the event $\{\bar a = \fg\}$. From \eqref{eq:highaction}, we have
\[
r_{t+1}^\fg-r_t^\fg=D_\fg(r_t^\fg),
\qquad
D_\fg(r)=\log\frac{1-F_\fg(-r)}{1-F_\fb(-r)},
\]
and $r_t^\fg\to\infty$ (Lemma \ref{lemma:rg}).  The LLR identity gives $F_\fg(-r)\leq\ee^{-r}F_\fb(-r)$,
so $F_\fg(-r)/F_\fb(-r)\to0$ as $r\to\infty$.
Using $\log(1+x)\sim x$ as $x\to0$, we obtain
\[
D_\fg(r)
=\log\left(
1+\frac{F_\fb(-r)-F_\fg(-r)}{1-F_\fb(-r)}
\right)
\sim F_\fb(-r)\to 0.
\]

By \eqref{eq:gaussian_private}, the private LLR is
$\ell=2s/\sigma^2$. Writing $\Phi$ for the standard normal
CDF, we obtain
\[
F_0(-r)
=\Prob_0\left[s\leq-(\sigma^2/2) r\right]
=\Phi\left(-\frac{\sigma r}{2}-\frac{m_0}{\sigma}\right),
\qquad
F_\fb(-r)
=\Phi\left(-\frac{\sigma r}{2}+\frac1\sigma\right).
\]
Using $\Phi(-z)\sim\ee^{-z^2/2}/(z\sqrt{2\pi})$ as
$z\to\infty$ gives
\begin{align}\label{eq:gaussian_tail_ratio}
\frac{F_0(-r)}{F_\fb(-r)}
& \sim \frac{\sigma r/2-1/\sigma}{\sigma r/2+m_0/\sigma}
\exp\left(
-\frac{(\sigma r/2+m_0/\sigma)^2-(\sigma r/2-1/\sigma)^2}{2}
\right) \nonumber \\
& \sim \exp\left(\frac{1-m_0^2}{2\sigma^2}\right)
\exp\left(-\frac{1+m_0}{2}r\right).
\end{align}

Let $c=(1+m_0)/2>0$. By \eqref{eq:gaussian_tail_ratio}
and $D_\fg(r)\sim F_\fb(-r)$, there exist $C>0$ and $T$
such that, for all $t\geq T$,
\[
F_0(-r_t^\fg) \leq C\ee^{-cr_t^\fg}D_\fg(r_t^\fg)=C\ee^{-cr_t^\fg}(r_{t+1}^\fg-r_t^\fg).
\]
Moreover, since
$r_{t+1}^\fg-r_t^\fg=D_\fg(r_t^\fg)\to0$,
one can choose $T$ large enough so that $r_{t+1}^\fg-r_t^\fg\leq1$
for $t\geq T$. For $x\in[r_t^\fg,r_{t+1}^\fg]$, this
implies $\ee^{-cr_t^\fg}\leq\ee^c\ee^{-cx}$, and therefore
\[
\ee^{-cr_t^\fg}(r_{t+1}^\fg-r_t^\fg)
\leq \ee^c\int_{r_t^\fg}^{r_{t+1}^\fg}\ee^{-cx}\,dx.
\]
Consequently,
\[
\sum_{t=T}^{\infty}F_0(-r_t^\fg)\leq C\ee^c
\sum_{t=T}^{\infty}
\int_{r_t^\fg}^{r_{t+1}^\fg}\ee^{-cx}\,dx =C\ee^c\int_{r_T^\fg}^{\infty}\ee^{-cx}\,dx
<\infty.
\]
Again, by the same logic that we use to deduce \eqref{eq:finitesum}, this is equivalent to $\Prob_0[\bar{a} = \fg] >0$. 
By symmetry,
$r_t^\fb=-r_t^\fg$ and $1-F_\fg(r) = F_\fb(-r)$. The corresponding Gaussian right-tail comparison gives
$$
\frac{1-F_0(r)}{1-F_\fg(r)}
\sim
\exp\left(\frac{1-m_0^2}{2\sigma^2}\right)
\exp\left(-\frac{1-m_0}{2}r\right).
$$
Since $(1-m_0)/2>0$, the same increment and integral
comparison yields
$\sum_{t \geq T}(1-F_0(r_t^\fg))<\infty$, and hence
$\Prob_0[\bar a=\fb]>0$.
Thus both immediate-agreement events have positive probability under $\Prob_0$, and Proposition \ref{cor:immd_learn_uniform} implies that asymptotic learning of informativeness fails.
\end{proof}

\section{Non-Uniform Prior} In this section, we prove Proposition \ref{prop:nonuniform_fails}. We relax the uniform prior assumption on $\theta$ while maintaining the independence assumption between $\theta$ and $\omega$.

Recall that $p_t= \Prob[\theta=\fg|H_t]$ and $q_t=\Prob[\omega=1|H_t]$ is the outside observer's belief that the source is informative at time $t$. Given that the prior on $\theta$ is not uniform, each agent $t$ chooses an action according to \eqref{eq:optimal_action_gen}. It immediately follows that $p_t$ also satisfies the overturning principle as in Lemma \ref{claim:overturning}. The following lemma is useful in proving Proposition \ref{prop:nonuniform_fails}. Let $a^{d}$ denote the default action that is favored by the non-uniform prior.

\begin{lemma}\label{lem:nonuniform}
Given a non-uniform prior on $\theta$, the following statements hold:
\begin{compactenum}[(i)]
    \item Conditioned on $\lim_{t\to \infty} q_t = 0$ almost surely, actions almost surely converge to $a^{d}$. 
    \item Conditioned on $\lim_{t\to \infty} q_t = 1$ almost surely, actions almost surely converge to $\theta$. 
\end{compactenum}
\end{lemma}
\begin{proof}
Fix a prior $\pi = \Prob[\theta=\fg] >1/2$, and so $a^d=\fg$. For part (i), suppose towards a contradiction that actions do not converge to $\fg$ almost surely. There are two possible cases: either $\lim_{t \to \infty} a_t = \fb$, or the actions never converge, i.e., $S = \infty$. By the overturning principle, these scenarios imply that either $\lim_{t\to \infty} p_t = p_\infty \leq 1/2$, or $p_\infty=1/2$. Together they imply that the event $\{p_\infty \leq 1/2\}$ occurs with positive probability. Since $\lim_{t\to \infty} q_t = q_\infty = 0$ almost surely, it follows from \eqref{eq:standardLearning} that
    \begin{align*}
      p_\infty  & = \Prob_0[\theta=\fg|H_\infty] = \pi > 1/2 \quad \text{almost surely,} 
    \end{align*} 
   a contradiction to the event $\{p_\infty \leq 1/2\}$ occurring with positive probability. Hence, we conclude that $\lim_{t \to \infty} a_t =\fg$ almost surely.

For part (ii), consider the case where $\theta= \fg$. Assume by contradiction, that either $\lim_{t \to \infty} a_t =\fb$ or $S= \infty$ occurs with positive probability. Again, by the overturning principle, these  imply that the event $\{p_\infty \leq 1/2\}$ occurs with positive probability. Since $\lim_{t\to \infty} q_t = 1$ almost surely, it follows from \eqref{eq:standardLearning} that $ \pi_\infty = p_\infty \leq 1/2$ with positive probability. By a standard martingale convergence argument with unbounded signals \citep{smith2000pathological}, we know that conditioned on $\theta=\fg$, $\pi_\infty = 1$ almost surely, which contradicts the event $\{\pi_\infty \leq 1/2\}$ occurring with positive probability. An analogous argument applies to the case where $\theta =\fb$. 
    \end{proof}

 \begin{proof}[Proof of Proposition \ref{prop:nonuniform_fails}] Suppose by contradiction that asymptotic learning of informativeness holds. By definition and Lemma \ref{lem:nonuniform}, we have that conditioned on $\omega=0$, $\lim_t a_t = a^d$ almost surely; similarly, conditioned on $\omega=1$, $\lim_t a_t = \theta$ almost surely. However, since $\Prob[\theta = a^d| \omega=1]>0$, there exists a history of actions at time infinity $H^d_{\infty}$ in which $\lim_t a_t = a^d$ and this history is possible under both probability measures $\Prob_0$ and $\Prob_1$. Following a similar argument to that used in the proof of Proposition \ref{thm:main} (for the only-if direction), this implies that asymptotic learning of informativeness fails, a contradiction.
\end{proof}

\bibliographystyle{ecta}
\bibliography{reference.bib}

\end{document}